\documentclass{article}

\usepackage{arxiv}
\usepackage{newtxtext}
\usepackage[utf8]{inputenc} 
\usepackage[T1]{fontenc}    
\usepackage{hyperref}       
\usepackage{url}            
\usepackage{booktabs}       
\usepackage{amsfonts}       
\usepackage{nicefrac}       
\usepackage{microtype}      
\usepackage{lipsum}         
\usepackage{graphicx}
\usepackage{natbib}
\usepackage{doi}

\usepackage{amsmath} 
\usepackage{cleveref} 
\usepackage{mathtools}
\usepackage{array}
\usepackage{algorithm}
\usepackage[caption=false,font=normalsize,labelfont=sf,textfont=sf]{subfig}
\usepackage{textcomp}
\usepackage{url}
\usepackage{verbatim}
\usepackage{mathtools}
\usepackage{amssymb}
\usepackage{upgreek}
\usepackage{xcolor}
\usepackage{amsthm}
\usepackage{enumerate}

\newtheorem{lem}{Lemma}
\newtheorem{remark}{Remark}
\newtheorem{theorem}{Theorem}
\newtheorem{proposition}{Proposition}

\title{A Perception-feedback position-tracking control for quadrotors}

\newif\ifuniqueAffiliation
\uniqueAffiliationtrue

\ifuniqueAffiliation 
\author{ \href{https://orcid.org/0000-0002-8931-3841}{\includegraphics[scale=0.06]{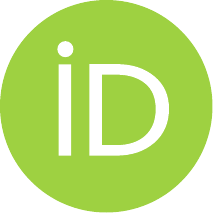}\hspace{1mm}Eduardo~Esp\'indola} \\
	School of Information Science and Engineering\\
	NingboTech University\\
	Ningbo 315100, Zhejiang, China \\
	\texttt{eespindola@zju.edu.cn} \\
	\And
	\href{https://orcid.org/0000-0003-3305-3179}{\includegraphics[scale=0.06]{orcid.pdf}\hspace{1mm}Yu~Tang}\thanks{Corresponding author.} \\
	School of Information Science and Engineering\\
	NingboTech University\\
	Ningbo 315100, Zhejiang, China \\
	\texttt{tangy@nbt.edu.cn} \\
}
\else
\usepackage{authblk}

\newbox{\orcid}\sbox{\orcid}{\includegraphics[scale=0.06]{orcid.pdf}} 
\author[1]{%
	\href{https://orcid.org/0000-0000-0000-0000}{\usebox{\orcid}\hspace{1mm}David S.~Hippocampus\thanks{\texttt{hippo@cs.cranberry-lemon.edu}}}%
}
\author[1,2]{%
	\href{https://orcid.org/0000-0000-0000-0000}{\usebox{\orcid}\hspace{1mm}Elias D.~Striatum\thanks{\texttt{stariate@ee.mount-sheikh.edu}}}%
}
\affil[1]{Department of Computer Science, Cranberry-Lemon University, Pittsburgh, PA 15213}
\affil[2]{Department of Electrical Engineering, Mount-Sheikh University, Santa Narimana, Levand}
\fi

\renewcommand{\shorttitle}{Perception-feedback position-tracking control}

\begin{document}
\maketitle

\begin{abstract}
	In this paper a position-tracking controller for quadrotors based on perception feedback is developed, which directly uses measurements from onboard sensors such as low cost IMUs and GPS to generate the control commands without state estimation.  Bias in gyros sensors are corrected to enhance the tracking performance. Practical stability of the origin of the tracking error system in the presence of external disturbances is proved using the Lyapunov analysis, which turns out to exponential stability in the absence of external disturbances. Numerical simulations are included to illustrate the proposed control scheme and to verify the robustness of the proposed controller under noisy measurements and parameter uncertainties.
\end{abstract}


\section{Introduction}\label{sec1}

Unmanned aerial vehicles (UAVs), in particular quadrotors, have attracted increasing attention for their versatility in a wide range of applications \citep{mahony2012multirotor,grzonka2011fully,serra2016landing} and the challenges posed in theory due to the nature of subactuation and the topological constraints on the rotation group \citep{hua2013introduction,chen2016novel,liu2016robust,zuo2014adaptive,zhao2014nonlinear,yu2016global}. Existing methods use the vector thrust approach to address the subactuation issue, where the propulsive thrust vector in the body-fixed frame is first determined through a desired attitude by an outer position controller, then the desired attitude is tracked through the applied torque by an inner rotation controller \citep{hua2009control,roberts2010adaptive,lee2013nonlinear}.  
To avoid the singularity issue in the position controller, the saturation control technique is commonly employed \citep{lee2013nonlinear}. The stability of the overall system is established by the Lyapunov analysis of cascade systems \citep{bertrand2011hierarchical,wang2014trajectory,roza2014class}.

Most existing work assumes that attitude measurements are available in terms of the rotation matrix \citep{lee2013nonlinear,roza2014class} and unit quaternions 
\citep{abdessameud2010global,guerrero2011bounded,xian2015nonlinear}
or Euler angles \citep{hamel2002dynamic,castillo2004real}. However, the attitude is not measured directly but is only reconstructed by an observer from the inertial measurements provided by IMUs or cameras \citep{rehbinder2003pose,mahony2008nonlinear}. In addition, when a dynamic process is used to reconstruct the attitude, the overall system must be analyzed to ensure stability, since the separation principle generally does not hold for nonlinear systems. Despite the importance of this issue, less attention has been paid in the literature, and only a few works have been reported on position regulation \citep{tayebi2013inertial,roberts2013new} using vector measurements.

Indeed, learning-based methods have sprung up for the design of perception-feedback controllers (cf. \cite{hwangbo2017control} and references therein). However, these methods require a large amount of data for training, which is expensive. In addition, parameters in a neural network may be quite large, making it difficult to solve most learning problems, and the process of obtaining the final control command lacks transparency \citep{kratsios2025generative}.  On the other hand, first principle models have been well developed for UAVs; controllers properly designed based on these models may have features such as guaranteed stability and robustness, easy implementation, and fewer on-line computational demands; however, they can perform poorly in the presence of model uncertainties and external disturbances \citep{liu2016robust}. Therefore, a combined approach to leverage the salient features of both model-based and learning-based approaches to address challenges in real applications is needed \citep{song2023reaching}.

This work presents an attempt to reduce the gap between these two approaches by designing a position-tracking controller for quadrotors that directly uses measurements from low-cost sensors onboard, such as IMUs and GPS, without state estimation. Following the vector thrust approach, the controller consists of a position feedback loop and an attitude feedback loop. 
The former determines the vector thrust to follow a desired position trajectory, while the latter is designed to track the desired attitude required for position control based on the method proposed in \cite{espindola2023attitude}, which uses the alignment errors between the vector measurements and the desired directional vectors for feedback, bypassing the need to estimate the attitude. Therefore, the proposed position-tracking controller requires only position measurements and vector measurements of at least two known inertial reference vectors. These measurements can be acquired by low-cost sensors such as GPS, IMUs or CCD cameras. Because the angular velocity measurement is biased in low-cost gyros, the bias is corrected in the control synthesis to enhance tracking performance. The overall system is analyzed using Lyapunov stability theory, and the practical stability (the uniform ultimate boundedness of solutions) of the tracking error system is proven, which implies exponential stability in the absence of external disturbances. To the best of our knowledge, this is the first result reported in the literature for trajectory tracking in quadrotors using measurements from sensors directly with exponential convergence in the presence of gyro bias and external disturbances. As a by-product, a technical result is proven regarding an upper bound on the attitude error in terms of the vector alignment error stated in Lemma \ref{lem-Alignment}-(iv), which allows for the establishment of exponential stability when there are no external disturbances. 

The main contribution is, therefore, the design of a perception-based position tracking controller for quadrotors using directly the measurements from onboard sensors without state estimation. 
The main features can be stated as follows: (1) by bypassing the attitude estimation, the overall system suppresses critical points introduced by an attitude observer, allowing for a simpler analysis and practical issues when the initial conditions are near one of these critical points \citep{lee2015global}; (2) unknown constant bias in the angular velocity by a gyro is corrected, allowing for high tracking precision in low-cost applications; and (3) enhanced robustness to uncertainties and external disturbances is achieved by exponential convergence around the desired equilibrium.

The remainder of the paper is organized as follows. Section \ref{Sec:Prel} provides preliminaries, including the notation, equations of rotational and translational motion of a UAV, and the description of onboard measurements. In Section \ref{Sec:CtrlProb}, the tracking problem is formulated, and the desired attitude is defined in terms of the thrust vector and the desired yaw angle. In Section \ref{Sec:TrackCtrl}, the position control is designed first, followed by the attitude control, which is devised using vector measurements and a gyro-bias observer.
The attitude tracking objective is stated as an alignment problem between the vector measurements of at least two known non-collinear inertial reference vectors and the desired directional vectors determined by the desired attitude and the inertial reference vectors. Alignment errors are defined, and their properties, including the upper bound of the attitude error in terms of the vector alignment error (Lemma  \ref{lem-Alignment}), are discussed. 
are summarized. Finally, the overall system, consisting of the position controller, the attitude controller, and the gyro-bias observer, is analyzed together to ensure the almost semi-global exponential stability of a residual set of the origin in the overall system. The residual set reduces to the origin if external disturbances are absent.   Section \ref{Sec:Sim} presents numerical simulations to illustrate the theoretical results and verify the robustness of the tracking controller under noisy measurements and parametric uncertainty. Section \ref{Sec:Conc} provides conclusions. The appendix 
gives some results tailored to the proof of technical lemmas.

\section{Preliminaries}\label{Sec:Prel}

\subsection{Notation}
The vector norm $\| u\|=(u^T u)^{1/2}$ and the matrix norm $\|  A\|=\lambda_{\max}^{1/2}(A^T A)$, $\forall u\in \mathbb R^3$ and $\forall A \in \mathbb{R}^{n\times m}$ are used. The Frobenius norm is $\|A\|_{F} = \mathrm{tr}^{1/2}(A^{T}A), \forall A\in \mathbb R^{n\times n}$, with $\mathrm{tr}(\cdot)$ denoting the matrix trace. If the matrix $A$ is symmetric,  $A=A^{T}>0$ indicates a positive definite matrix, with $\lambda_{\max}(A)$ and $\lambda_{\min}(A)$ being its maximum and minimum eigenvalues, respectively. $I_n$ denotes the identity matrix of size $n\times n$, and $0_{n\times m}$ denotes a zero matrix of size $n\times m$. The {\it wedge} map $(\cdot)^{\wedge} \vcentcolon \mathbb{R}^{3} \to \mathfrak{so}(3)$ is an isomorphism between $\mathbb{R}^{3}$ and the space of skew-symmetric matrices $\mathfrak{so}(3)\vcentcolon = \{ A\in\mathbb{R}^{3\times 3} \mid A^{T}=-A \}$, whose inverse is denoted by the {\it vee} map $(\cdot)^{\vee} \vcentcolon \mathfrak{so}(3)\to \mathbb{R}^{3}$. The cross product is $u^{\wedge}v=u\times v, \ \forall u,v\in \mathbb R^3$. A unit sphere of dimension $n-1$ embedded in space $\mathbb{R}^{n}$ is denoted by $\mathcal{S}^{n-1}\vcentcolon = \{ u\in\mathbb{R}^{n} \mid u^{T}u= 1 \}$, while a $\mathbb{R}^{n}$-embedded ball of radius $r$ is expressed as $B_{r}\vcentcolon =\left\{ u\in\mathbb{R}^{n}\mid \|u\|\leq r \right\}$.

\subsection{Quadrotor model}
Consider the body frame $\mathbf{B}=\left\{ O_{B}, \{ e_{1} ,e_{2},e_{3} \} \right\}$ with the origin $O_{B}$ attached to the center of mass of the rigid body and an inertial frame $\mathbf{I} = \left\{ O_{I}, \{e_{x} ,e_{y},e_{z} \} \right\}$ as shown in Fig. \ref{fig:Quadrotor}. The attitude of the body frame relative to the inertial frame is then fully defined by the rotation matrix $R\in SO(3)\vcentcolon= \left\{ R\in\mathbb{R}^{3\times 3} \mid R^{T}R = RR^{T} = I_{3} , \mathrm{det}(R)=1  \right\}$. 

\graphicspath{ {./Figs/} }
\begin{figure}[h]
\centering
\includegraphics[scale=0.7]{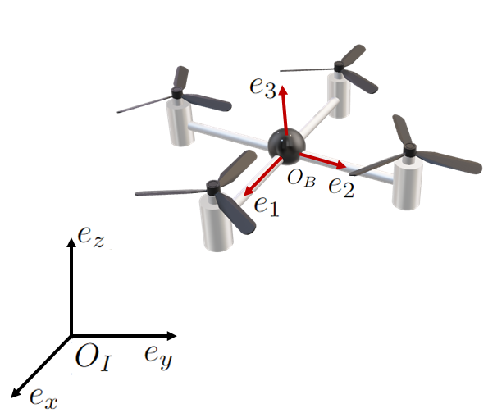}
\caption{Reference frames of the quadrotor.}
\label{fig:Quadrotor}
\end{figure}

The motion equations of a quadrotor UAV are given by 
\begin{align}
    \dot{x} &= v, \label{eq:KinX}\\
    m\dot{v} &= -mge_{z} + fRe_{z} +f_{d}, \label{eq:DynV}\\
    \dot{R} &= R(\omega)^{\wedge}, \label{eq:KinR} \\
    M\dot{\omega} &= (M\omega)^{\wedge}\omega + \tau +\tau_{d} , \label{eq:DynOm}
\end{align}
where $(x,v)\in\mathbb{R}^{3}\times\mathbb{R}^{3}$ is the position and linear velocity of origin $O_{B}$ expressed in the inertial frame, $(R,\omega)\in SO(3)\times \mathbb{R}^{3}$ is the attitude and body-fixed angular velocity of the quadrotor, $(f,\tau)\in \mathbb{R}\times \mathbb{R}^{3}$ is the thrust force and the control torque expressed in the body-fixed frame $\mathbf{B}$, and $(f_{d},\tau_{d})\in \mathbb{R}^{3}\times \mathbb{R}^{3}$ represents bounded external disturbances, such as aerodynamic effects, unmodeled dynamics,  and torque disturbances. Scalar $m>0$ denotes the mass of the quadrotor, $M=M^{T}\in\mathbb{R}^{3\times 3}$ denotes the positive definite inertia matrix, and $e_{z}=[0,0,1]^{T}$ represents the direction of the $z$ -axis of the inertial frame, expressed in the body-fixed frame $\mathbf{B}$.

\subsection{Measurements}
The quadrotor is assumed to be equipped with sensors that measure the state $(x,v)$ expressed in the frame $\mathbf{I}$, such as a GPS sensor in outdoor applications or CCD cameras for indoor applications. In addition, the vehicle has access to the vector measurements  $v_i\in \mathcal S^2$  in the body frame of $n\geq 2$ known constant reference vectors  $r_i\in \mathcal S^2$  in the inertial frame as
\begin{eqnarray}\label{eq:vi}
v_i = R^{T} r_{i}, \quad i=1,2,\cdots , n.
\end{eqnarray}
 Inertial measurement units (IMU) and CCD cameras are typical low-cost sensors used in UAVs that provide vector measurements. 

The angular velocity is provided by a gyroscope. MEMS-technology-based gyroscopes have a relatively large bias compared to high-precision gyroscopes, which must be corrected for control purposes. Therefore  angular-velocity measurement $\omega_g$ is modeled as follows
\begin{equation}\label{eq:VelM}
 \omega_{g} = \omega + b,
\end{equation}
where $b\in\mathbb{R}^{3}$ denotes the gyro bias.

The following assumptions are made for the control design.
\begin{itemize}
\item Assumption A1: Among the $n$ known constant inertial reference vectors $r_i\in \mathcal S^2$, there are at least two noncollinear vectors.
\item Assumption A2: Gyro bias $b$ is a constant unknown vector that satisfies $\|b\|\leq \mu_{b}$, for some known $\mu_b>0$. 
\end{itemize}

\section{Control problem statement}\label{Sec:CtrlProb}

Given a desired position trajectory $x_{d}(t)\in \mathbb R^3$ and a desired yaw angle $\psi_d(t) \in\mathbb R$, which are assumed to be sufficiently smooth, define the position-tracking errors 
\begin{align}
    \tilde{x} &= x- x_{d}, \label{eq:xTld}\\
    \tilde{v} &= v-\dot{x}_{d}.  \label{eq:vTld} 
\end{align}
By the translational motion equations \eqref{eq:KinX}-\eqref{eq:DynV} the error dynamics is 
\begin{align}
    \dot{\tilde{x}} &= \tilde{v}, \label{eq:xTldP}\\
    m\dot{\tilde{v}} 
        &= -mge_{z}- m\ddot{x}_{d}+T + f\left( R- R_{d}\right)e_{z}+f_d, \label{eq:vTldP} 
\end{align}
where $R_{d}\in SO(3)$ is the desired attitude to be determined and $T:=fR_d e_{z}\in \mathbb{R}^{3}$ is the thrust vector.
Note that the unit vector $Re_{z}\in\mathcal{S}^{2}$ in $T$ gives the thrust direction to guide the quadrotor to follow the desired position trajectory, while the thrust force $f=\|T\|$ provides the thrust magnitude. In addition, the thrust direction restricts the desired attitude, leaving the yaw angle $\psi\in\mathbb{R}$  as the only degree of freedom that can be controlled regardless of position. 
The desired attitude trajectory can be computed as \citep{lee2010geometric}
\begin{align}\label{eq:Rd}
    R_{d} &\vcentcolon = \left[ c_{1} \; c_{2} \; c_{3} \right], \quad \omega_{d} = \left(R^{T}_{d}\dot{R}_{d}\right)^{\vee}, 
\end{align}
where $c_1,\ c_2,\ c_3\in\mathcal{S}^{2}$ are the column vectors of $R_{d}$ given by
\begin{equation*}
  c_{3}  = \frac{T}{\|T\|}, \ c_{2} = \frac{c^{\wedge}_{3}c_{d}}{\|c^{\wedge}_{3}c_{d}\|}  ,\  c_{1} = c^{\wedge}_{2}c_{3},
\end{equation*}
being $c_{d} = \left[ \cos{(\psi_{d}(t))},\sin{(\psi_{d}(t)),0} \right]^{T}$ the desired yaw direction. 

Thus, the control objective is to design a control law for the thrust vector $T$ which gives the thrust force $f=\|T\|$ and the torque $\tau$ such that $(\tilde{x},\ \tilde v) \to(0_{3\times 1},0_{3\times 1})$ and $(\psi-\psi_d, \dot\psi-\dot \psi_d)\to (0,\ 0)$ with all signals bounded.

\begin{remark}
It is clear that \eqref{eq:Rd} is valid only when $f=\|T\|> 0$ and $c_{3}$, $c_{d}$ are noncollinear. Both conditions can be ensured by designing a saturated control thrust vector $T$ with its third component larger than zero.  Therefore, the thrust vector $T$ must be devised to achieve both the tracking objective $(\tilde{x},\ \tilde v) \to (0_{3\times 1},0_{3\times 1})$ and the ensuring $\|T\|> 0$. In addition, the thrust vector must be designed with sufficient smoothness to ensure $T\in\mathcal{C}^{2}$ so that 
the desired angular velocity $\omega_{d}(T,\dot{T},\psi_{d},\dot{\psi}_{d}) = \left(R^{T}_{d}\dot{R}_{d}\right)^{\vee}$ and its first time derivative $\dot{\omega}_{d}(T,\dot{T},\ddot{T},\psi_{d},\dot{\psi}_{d},\ddot{\psi}_{d}) =  \left(R^{T}_{d}\ddot{R}_{d} + \dot{R}^{T}_{d}\dot{R}_{d}\right)^{\vee}$ are smooth.
\end{remark}

\section{Tracking controller design}\label{Sec:TrackCtrl}

\subsection{Positioning controller design}

The following control law is proposed
\begin{align}
    T&= mge_{z}+m\ddot{x}_{d} + m\left( kI_{3}+ k_{x}I_{3}+K_{f}\right)y, \label{eq:CtrlPs}\\
    \dot{e}_{f} &=\mathrm{Cosh}^{2}(e_{f})\Big( -K_{f}y + k^{2}_{x}\big( 1-\frac{1}{m}\big)\tilde{x} -k\eta \Big), \ \ e_{f}(0) = 0_{3\times 1},
   \label{eq:efp}\\
    \eta &= \tilde{v} + k_{x}\tilde{x} + y, \label{eq:eta}
\end{align}
where $y:=\mathrm{Tanh}(e_{f})$ is an auxiliary variable, $k>0$ and $k_{x}>0$ are scalar gains, $K_{f}=\text{diag}\{ k_{f,1}\ k_{f,2}\ k_{f,3}\}\in\mathbb{R}^{3\times 3}$, and the hyperbolic functions are given by
\begin{align*}
    \mathrm{Tanh}(u) &= \left[ \tanh{(u_{1})},\tanh{(u_{2})},\tanh{(u_{3})} \right]^{T}\in\mathbb{R}^{3},\\
    \mathrm{Cosh}(u) &= \mathrm{diag}\left\{ \cosh{(u_{1})},\cosh{(u_{2})},\cosh{(u_{3})} \right\}\in\mathbb{R}^{3\times 3},\\
    \mathrm{Sech}(u) &= \mathrm{diag}\left\{ \mathrm{sech}(u_{1}),\mathrm{sech}(u_{2}),\mathrm{sech}(u_{3}) \right\}\in\mathbb{R}^{3\times 3},
\end{align*}
for $u=[u_1,u_2,u_3]^T\in\mathbb{R}^{3}$.

The following assumption is made.
\begin{itemize}
\item Assumption A3: The desired acceleration along the $z$-axis $\ddot{x}_{d,z}$ satisfies $|\ddot{x}_{d,z}|\leq \mu_{d}$ for some $\mu_d<g$.
\end{itemize}
This assumption is a realistic operating condition for most quadrotors. Therefore, if the control gains are chosen such that $k+ k_{x}+k_{f,3} <g -\mu_{d}$ then
\begin{align}\label{eq:TCond}
      0<g -\mu_{d}-(k+ k_{x}+k_{f,3})\leq \frac{\|T\|}{m}\leq g+\mu_d+k+ k_{x}+\lambda_{max}(K_f):=\frac{f_{max}}{m}.
\end{align}

The following proposition establishes the practical stability of the origin $(\tilde{x},\tilde{v})=0_{6\times 1}$ of the translational error dynamics \eqref{eq:xTldP}-\eqref{eq:vTldP} .

\begin{proposition}[\bf{Practical stability of position subsystem}]\label{thm1}
Consider the dynamics of translational errors \eqref{eq:xTldP}-\eqref{eq:vTldP} in closed loop with the control law \eqref{eq:CtrlPs}-\eqref{eq:eta}. Under Assumption A3, let the control gains be such that $0<k_x<g-\mu_d$, $k>k_x$, and $\lambda_{\min}(K_f)>\frac{k_x}{4m^2}$.
Then the origin $(\tilde{x},\tilde{v})=0_{6\times 1}$ is practically stable. 
\end{proposition}
\begin{proof}
The dynamics of position errors \eqref{eq:xTldP}-\eqref{eq:vTldP} in closed loop with the control law \eqref{eq:CtrlPs}-\eqref{eq:eta} is
\begin{align}
    \dot{\tilde{x}} &= \tilde{v} = \eta - k_{x}\tilde{x} -y, \label{eq:xTldPeta}\\
    \dot{\tilde v} &=(kI_3+k_x I_3+K_f)y +\frac{1}{m}\big(f(R-R_d)e_z+f_d\big). \label{eq:vTldPeta}
\end{align}
The dynamics of $y=\mathrm{Tanh}(e_{f})$ can be calculated using  $\frac{d}{dt}\mathrm{Tanh}(e_{f}) = \mathrm{Sech}^{2}(e_{f})\dot{e}_{f} $ and \eqref{eq:efp} as
\begin{equation}
    \dot y=
     -K_{f}y + k^{2}_{x}\left( 1-\frac{1}{m}\right)\tilde{x} -mk\eta.\label{eq:TefP}
\end{equation}
At last, taking  time derivative of \eqref{eq:eta} and substituting \eqref{eq:xTldPeta}-\eqref{eq:TefP} yields
\begin{align}
    \dot{\eta} &= \dot{\tilde{v}} + k_{x}\tilde{v} + \mathrm{Sech}^{2}(e_{f})\dot{e}_{f} \notag \\
    &=  (kI_3+k_x I_3+K_f)y +\frac{1}{m}\big(f(R-R_d)e_z+f_d\big) +k_{x}\left( \eta - k_{x}\tilde{x} -y\right) 
    +\left( -K_{f}y + k^{2}_{x}( 1-\frac{1}{m})\tilde{x} -k\eta \right) \notag \\
    &= -\big( mk -k_{x}\big) \eta+ky - \frac{k^{2}_{x}}{m}\tilde{x}+\frac{1}{m}\big(f(R-R_d)e_z+f_d\big). \label{eq:etaTldP}
\end{align}
Let $\zeta_{1} \vcentcolon = \left[ \eta^{T} \; \tilde{x}^{T}\; y^T \right]^{T}\in\mathbb{R}^{9}$ be the state of the position error dynamics in closed loop with the controller \eqref{eq:CtrlPs}-\eqref{eq:eta},  which has the origin $\zeta_{1}=0_{9\times 1}$ as unique equilibrium when $R=R_d$ and $f_d=0_{3\times 0}$. To study its stability, consider the following positive definite function
\begin{equation}\label{eq:V1}
V_{1}(\zeta_{1}) \vcentcolon = \frac{m}{2} \|\eta\|^{2} + \frac{k^{2}_{x}}{2}\|\tilde{x}\|^{2}+\frac{1}{2}\|y\|^{2} ,
\end{equation}
which is radially unbounded for all $\zeta_{1}\in\mathbb{R}^{9}$ and satisfies
\begin{align}
\frac{1}{2}\min\left\{  m,k^{2}_{x},1\right\}\|\zeta_{1}\|^{2}\leq &V_{1}(\zeta_{1}) \leq \frac{1}{2} \max\left\{ m,k^{2}_{x},1\right\}\|\zeta_{1}\|^{2}.
\end{align}
The time evolution of \eqref{eq:V1} along the error dynamics \eqref{eq:xTldPeta} and \eqref{eq:TefP}-\eqref{eq:etaTldP} is given by
\begin{align}\label{eq:V1p}
    \dot{V}_{1} &= m\eta^{T}\dot{\eta} + k^{2}_{x}\tilde{x}^{T}\dot{\tilde{x}} + y^{T}\dot y \notag\\
    &= \eta^{T}\Big( -\big( mk -k_{x}\big) \eta+ky - \frac{k^{2}_{x}}{m}\tilde{x}+\frac{1}{m}\big(f(R-R_d)e_z+f_d\big)\Big)  + k^{2}_{x}\tilde{x}^{T}\big( \eta - k_{x}\tilde{x} -y\big)  \notag\\
    &\quad + y^T\Big( -K_{f}y + k^{2}_{x}\big( 1-\frac{1}{m}\big)\tilde{x} -mk\eta \Big) \notag\\
    &= -m( mk -k_{x})\eta^{T} \eta - k^{3}_{x}\tilde{x}^{T}\tilde{x} -y^{T}K_{f}y - \frac{k^{2}_{x}}{m}y^T\tilde{x}+ \eta^T \big(f(R-R_d)e_z+f_d\big) \notag\\
    &\leq - \Big[\|\eta\| \ \|\tilde x\| \ \|y\|\Big]^TQ_1 \Big[\|\eta\| \ \|\tilde x\|  \ \|y\|\Big]
    +\|\eta\|\Big(f\|R-R_d\|+\|f_d\|\Big)  
\end{align}
where
\begin{equation}\label{eq:Q1}
    Q_{1} = \left[ \begin{array}{ccc}
        m(k-k_{x}) & 0 & 0 \\
         0 & k^{3}_{x} & -\frac{k^{2}_{x}}{2m} \\
         0 & -\frac{k^{2}_{x}}{2m} & \lambda_{min}(K_{f})
    \end{array}\right] \in \mathbb{R}^{3\times 3},
\end{equation}
which is positive definite if the gains are such as $0<k_x<g-\mu_d$, $k>k_x$, and $\lambda_{\min}(K_f)>\frac{k_x}{4m^2}$. Since $f =\|T\|\leq f_{max}$ by \eqref{eq:TCond}, the term $f\|R-R_d\|+\|f_d\| $ is bounded. Thus, the origin $\zeta_{1}=0_{9\times 1}$ is practically stable by Lemma \ref{lem-pracStab}-(ii). The attitude error $\|R-R_d\|$ can be rendered to converge to a residual set of zero by the attitude tracking control designed in the sequel.
\end{proof}

\subsection{Attitude controller design}
\noindent 
{\it Alignment errors.} 
The attitude controller is developed using vector and gyro measurements. Towards this purpose, let the desired direction be defined as follows:
\begin{equation}\label{eq:vdi}
    v_{d,i} = R^{T}_{d}r_{i}, \quad i=1,2,\cdots ,n. 
\end{equation}
Then, under Assumption A1 the convergence $v_{i}\to v_{d,i}$,  $i=1,2,\cdots , n$, ensures $R\to R_{d}$. Thus, the attitude control problem becomes a vector alignment problem. The alignment error can be measured by the following variables
\begin{align}
    \varepsilon (t) &= \sum^{n}_{i=1}k_{i}\left(1-v^{T}_{i}v_{d,i}\right) ,\label{eq:varEps}\\
    z(t) &= \sum^{n}_{i=1}k_{i}v^{\wedge}_{i}v_{d,i}, \label{eq:z}
\end{align}
where $k_{i}>0$ is the weight assigned to $i$th sensor according to its confidence level. The variable $\varepsilon\in\mathbb{R}$ relates the inner product between the actual direction $v_i$ and the desired direction $v_{d,i}$, and measures the deviation in both alignment and steering, while the variable $z\in\mathbb{R}^{3}$ obtained by the cross product measures only the misalignment. Note that $v_{i}=v_{d,i}$  yields $\varepsilon =0$ and $z=0_{3\times 1}$, while $v_{i}=-v_{d,i}$  gives $z=0_{3\times 1}$ and $\varepsilon = 2\sum^{n}_{i=1}k_{i}$. 

The dynamics of the alignment error variables \eqref{eq:varEps}-\eqref{eq:z} can be obtained by  \eqref{eq:KinR}, the desired attitude trajectory $\dot{R}_{d}=R_{d}(\omega_{d})^{\wedge}$, and $\dot{v}_{i}=v^{\wedge}_{i}\omega$ and $\dot{v}_{d,i}=v^{\wedge}_{d,i}\omega_{d}$ as follows
\begin{align}
    \dot{\varepsilon} &= z^{T}(\omega - \omega_{d}),\label{eq:varepsP} \\
    \dot{z} &= J(\omega - \omega_{d}) + z^{\wedge}\omega_{d}, \label{eq:zP}
    \end{align}
where $J\in\mathbb{R}^{3\times 3}$ is 
    \begin{align}
    J &\vcentcolon = \sum^{n}_{i=1}k_{i}\left(v^{\wedge}_{d,i}\right)^{T}v^{\wedge}_{i} , \label{eq:J}
\end{align}and is bounded by $\|J\|\leq \sum^{n}_{i=1}k_{i}$.

The following lemma relates these error variables with the attitude error defined below and gives an upper bound on the term $\|R-R_d\|$ in \eqref{eq:vTldP}, which is instrumental for subsequent control designs.

\begin{lem}[\bf{Alignment error variables $\varepsilon (t)$ and $z(t)$}]\label{lem-Alignment}\hfill 
\begin{enumerate}[(i)]
    \item Let the attitude error be defined as $\tilde{R}\vcentcolon =RR^{T}_{d}\in SO(3)$. Then $z=0_{3\times 1}$ implies $\tilde{R}=I_{3}$ or $\tilde{R}=R_{j}\vcentcolon = I_{3} + 2(v^{\wedge}_{w_j})^{2}$, for $j=1,2,3$, where $v_{w_j}\in\mathcal{S}^{2}$ are unit eigenvectors of the symmetric positive definite matrix \citep{tayebi2013inertial}
    \begin{equation}\label{eq:W}
        W \vcentcolon = -\sum^{n}_{i=1}k_{i}(r^{\wedge}_{i})^{2},
    \end{equation}
    with the associated eigenvalues 
    ordered, without loss of generality,  as $\lambda_{w,1}\geq \lambda_{w,2} \geq \lambda_{w,3}$.
    \item The variable $\varepsilon (t)$ satisfies $0\leq \varepsilon \leq 2\sum^{n}_{i=1}k_{i}$.
    \item For any $\alpha_1>0$ there exists $\beta>0$, such that
    \begin{equation}\label{eq:bcond}
      \alpha_{1}\varepsilon \leq \frac{\beta}{2}\| z \|^{2}, \ \ \forall t\geq 0 ,
    \end{equation}
    $\forall \tilde{R} \in  SO(3) \backslash \mathcal{B}_{\epsilon}$, where $\mathcal{B}_{\epsilon}\vcentcolon = \mathcal{B}_{\epsilon ,1} \cup \mathcal{B}_{\epsilon ,2} \cup \mathcal{B}_{\epsilon ,3}$ with
        \begin{align}\label{eq:BallQ}
            \mathcal{B}_{\epsilon ,j} &\vcentcolon = \left\{ \tilde{R}\in SO(3) \mid  \tilde{R} = R_{\sigma ,j} , \forall \sigma\in [0,\epsilon ] \right\} ,\notag\\
            R_{\sigma ,j} & \vcentcolon = I_{3}-2\sigma \sqrt{1-\sigma^{2}}v^{\wedge}_{w_j} + 2\left( 1-\sigma^{2}\right)(v^{\wedge}_{w_j})^{2},
        \end{align}
    are closed balls centered at $R_{j}\in SO(3)$, and (arbitrarily small) radius $\epsilon_{j} >0$, for $j=1,2,3$. Furthermore, the constant $\beta$ is given by
    \begin{equation}\label{eq:beta}
        \beta \geq 2\frac{\lambda_{w,1}}{\epsilon^{2}_{j}\lambda^{2}_{w,3}}\alpha_{1} , \;\forall j=1,2,3.
    \end{equation}
    \item Let $\bar{W}\in\mathbb{R}^{3\times 3}$ be defined as 
    \begin{equation}\label{eq:Wb}
        \bar{W} \vcentcolon = \sum^{n}_{i=1}k_{i}r_{i}r^{T}_{i},
    \end{equation}
    which is positive definite under Assumption A1 (cf. Lemma 2, \cite{tayebi2013inertial}). Let $\varpi \vcentcolon = \mathrm{tr}\left( \bar{W}^{-1}\right)>0$, then the following condition holds for all $\tilde{R}\in SO(3)\backslash \mathcal{B}_{\epsilon}$
    \begin{equation}\label{eq:RRd}
        \|R-R_{d}\| \leq \sqrt{\frac{\varpi \beta}{\alpha_{1}}}\|z\|,\quad \forall t\geq 0.
    \end{equation}
\end{enumerate}
\end{lem}
\begin{proof}
The proof of items (1)-(3) can be found in \cite{espindola2023attitude}. The proof of item (4) is provided in the appendix. 
\end{proof}
 
\noindent
{\it Attitude tracking controller.} Given the desired attitude trajectory $(R_d,\ \omega_d)\in SO(3)\times \mathbb{R}^{3}$ with $(\omega_{d},\dot{\omega}_{d})$ continuous and bounded, the following attitude controller is proposed.
\begin{align}
    \tau &= M\dot{\hat{\omega}}_{r} -(M\hat{\omega})^{\wedge}\omega_{r} - K_{c}\left(\hat{\omega} - \omega_{r}\right) - \left( \alpha_{1}I_{3}+\alpha_{2}J^{T}\right) z,\label{eq:AttCtrl}\\
    \dot{\hat{\omega}}_{r} &= -\lambda_{c}J\left( \hat{\omega}-\omega_{d}\right) - \lambda_{c}z^{\wedge}\omega_{d} + \dot{\omega}_{d} ,\label{eq:omgREstp}\\
    \omega_{r} &= -\lambda_{c}z + \omega_{d},\label{eq:omgR}
\end{align}
where $\lambda_{c}>0$, $\alpha_{1,2}>0$  and $0<K_{c}\in\mathbb{R}^{3\times 3}$ are controller gains.  The angular velocity estimation  $\hat{\omega}: = \omega_{g} -\hat{b}$, with  $\hat b\in \mathbb R^3$ as the gyro bias estimation,  is  obtained by correcting the gyro bias $b\in\mathbb{R}^{3}$ in the measured gyro rate $\omega_g$ using the following gyro bias observer 
\begin{align}
    \hat{b} &= \bar{b} -\sum^{n}_{i=1}k_{i}(v^{\wedge}_{f,i})^{T}\Lambda_{i}v_{i}, \label{eq:bEst}\\
    \dot{\bar{b}} &= K_{v}\hat{\omega} +\gamma_{f}\sum^{n}_{i=1}k_{i}(\Lambda_{i}v_{i})^{\wedge}(v_{i}-v_{f,i}), \label{eq:bbarP}\\
    \dot{v}_{f,i} &= \gamma_{f}(v_{i} - v_{f,i}) ,\quad v_{f,i}(0)=v_{i}(0), \label{eq:vfp}
\end{align}
where $0<\Lambda_{i}=\Lambda_i^T\in\mathbb{R}^{3\times 3}$, $i=1,2,...,n$, are observer gains, $\gamma_{f}>0$ is the filter gain, and $K_{v}\in\mathbb{R}^{3\times 3}$ is defined as
\begin{equation}\label{eq:Kf}
    K_{v} = \sum^{n}_{i=1}k_{i}(v^{\wedge}_{f,i})^{T}\Lambda_{i}v^{\wedge}_{i}.
\end{equation}

\noindent
{\it Attitude error dynamics.}
By \eqref{eq:bEst}-\eqref{eq:vfp} the dynamics of the bias estimation error is  
\begin{equation}\label{eq:bTildP}
    \dot{\tilde{b}} = \dot{\hat{b}} - \dot{b} = -K_{v}\tilde{b}.
\end{equation}
Taking the time derivative of $\tilde{\omega}= \omega - \omega_{r}$ and substituting \eqref{eq:DynOm} in closed-loop with controller \eqref{eq:AttCtrl}-\eqref{eq:omgR}, it yields
\begin{equation}
    M\dot{\tilde{\omega}}
    = (M\omega)^{\wedge}\tilde{\omega} - K_{c}\tilde{\omega} +G\tilde{b} -\left( \alpha_{1}I_{3}+\alpha_{2}J^{T}\right) z .\label{eq:OmgErrP}
\end{equation}
The dynamics \eqref{eq:varepsP} and \eqref{eq:zP} are rewritten by adding and subtracting $\omega_{r}$ in $\omega$ as follows
\begin{align}
    \dot{z} &= J\tilde{\omega} -\lambda_{c}Jz + z^{\wedge}\omega_{d}, \label{eq:zP2} \\
    \dot{\varepsilon} &= z^{T}\tilde{\omega} -\lambda_{c}z^{T}z.\label{eq:varepsP2} 
\end{align}

Let $\zeta_{2}\vcentcolon = \left[ z^{T},\tilde{\omega}^{T},\tilde{b}^{T}  \right]^{T}\in\mathbb{R}^{9}$ be the state of the closed-loop attitude error dynamics described by \eqref{eq:bTildP}-\eqref{eq:zP2}. The following proposition studies the closed-loop stability of attitude error dynamics.

\begin{proposition}[\bf{Stability of the rotational error dynamics}]\label{thm2}
The controller \eqref{eq:AttCtrl}-\eqref{eq:omgR} with the bias observer \eqref{eq:bEst}-\eqref{eq:vfp} in closed loop with rotational error dynamics \eqref{eq:KinR}-\eqref{eq:DynOm} renders the origin  
$\zeta_{2}= \left[ z^{T},\tilde{\omega}^{T},\tilde{b}^{T}  \right]^{T}=0_{9\times 1}=0$  practically stable 
if the controller gains are chosen as $K_{c}=K^{T}_{c}>0$, $\lambda_{c}>0$, $\alpha_{1,2}>0$, and the gains in the observer $\gamma_{f}>0$, $k_{i}>0$, and $\Lambda_{i}=\Lambda^{T}_{i}>0$, for $i=1,2,\ldots ,n$ such that
\begin{align}
    \lambda_{o}&\vcentcolon = \lambda_{\min}(K_{o})-\epsilon_{v}\sum^{n}_{i=1}k_{i}\lambda_{\max}(\Lambda_{i}) > \frac{\|G\|^{2}}{4\lambda_{\min}(K_{c})} ,\label{eq:lo} \\
    \lambda_{a} &\vcentcolon = \alpha_{1} -\alpha_{2}\sum^{n}_{i=1}k_{i} >0 ,\label{eq:la}
\end{align}
where $\epsilon_v$ is an upper bound of the difference between the filter input and output, i.e., $\|v_{i}-v_{f,i}\| < \epsilon_{v}$, and can be made arbitrarily small provided that the filter gain is large enough (Lemma 3.1 of \cite{espindola2022new}), and 
\begin{align}
    K_{o} &\vcentcolon = \sum^{n}_{i=1}k_{i}(v^{\wedge}_{i})^{T}\Lambda_{i}v^{\wedge}_{i},\label{eq:Ko} \\
    G &\vcentcolon = K_{c} - \omega^{\wedge}_{r}M + \lambda_{c}MJ.\label{eq:G}
\end{align}
\end{proposition}
\begin{proof}


Consider the following positive definite function
\begin{equation}\label{eq:V2}
    V_{2}(\zeta_{2},t) \vcentcolon = \frac{1}{2}\tilde{\omega}^{T}M\tilde{\omega} + \frac{1}{2}\|\tilde{b}\|^{2} + \frac{\alpha_{2}}{2}\|z\|^{2} +\alpha_{1}\varepsilon,
\end{equation}
which is positive definite and radially unbounded.
In addition, according to Lemma \ref{lem-Alignment}-(iii), $V_{2}(\zeta_{2},t)$ satisfies 
\begin{equation}\label{eq:V2UpBound}
\frac{1}{2}\min\left\{ \lambda_{\min}(M),1,\alpha_{2} \right\} \|\zeta_{2}\|^{2} \leq V_{2}\leq \frac{1}{2}\max\left\{ \lambda_{\max}(M), 1,\alpha_{2}+\beta  \right\} \|\zeta_{2}\|^{2} .
\end{equation}
Thus, the time derivative of \eqref{eq:V2} along trajectories \eqref{eq:bTildP}-\eqref{eq:varepsP2} is 
\begin{align}\label{eq:V2P}
\dot{V}_{2} &= \tilde{\omega}^{T}M\dot{\tilde{\omega}}+\tilde{b}^{T}\dot{\tilde{b}} + \alpha_{2}z^{T}\dot{z}+\alpha_{1}\dot{\varepsilon} \notag\\
&= -\tilde{\omega}^{T}K_{c}\tilde{\omega} + \tilde{\omega}^{T}G\tilde{b} - \tilde{b}^{T}K_{f}\tilde{b} - \lambda_{c}z^{T}\left( \alpha_{1}I_{3}+\alpha_{2}J\right)z+\tilde\omega^T\tau_d \notag\\
&\leq - \Big[\|z\| \ \|\tilde \omega\| \ \|\tilde b\|\Big]^T Q_2 \Big[\|z\| \ \|\tilde \omega\| \ \|\tilde b\|\Big]+\|\tilde\omega\|\|\tau_d\|
\end{align}
where $Q_{2}\in\mathbb{R}^{3\times 3}$ is given by
\begin{equation}\label{eq:Q2}
    Q_{2}= \left[ \begin{array}{ccc}
         \lambda_{c}\lambda_{a}& 0 & 0  \\
          0 & \lambda_{\min}(K_{c}) & -\frac{1}{2}\|G\| \\
          0 & -\frac{1}{2}\|G\| & \lambda_{o}
    \end{array} \right], 
\end{equation}
which is positive definite under conditions \eqref{eq:lo}-\eqref{eq:la}. Therefore, the origin  $\zeta_{2}=0_{9\times 1}$ is practically stable by Lemma \ref{lem-pracStab}-(ii).

\end{proof}

\begin{remark}
In the absence of torque disturbance, i.e., $\tau_d=0_{3\times 1}$,   by Lemma \ref{lem-Alignment}-(i), convergence of $z\to 0_{3\times 1}$ implies that $\tilde{R}\to I_{3}$. Therefore, equilibrium $(\tilde{R},\tilde{\omega},\tilde{b})=(I_{3},0_{3\times 1},0_{3\times 1})$ is almost globally asymptotically stable for all $(\tilde{R}(0),\tilde{\omega}(0),\tilde{b}(0))\in\mathcal{X}_{a} \vcentcolon = SO(3)\backslash \{R_{j}\} \times \mathbb{R}^{3}\times\mathbb{R}^{3}$ and almost semiglobally exponentially stable for $(\tilde{R}(0),\tilde{\omega}(0),\tilde{b}(0))\in \mathcal{X}_{e}\vcentcolon = SO(3)\backslash \mathcal{B}_{\epsilon}\times \mathbb{R}^{3}\times \mathbb{R}^{3}$, where $\mathcal{B}_{\epsilon}$ is the union of closed balls given by \eqref{eq:BallQ} defined in Lemma \ref{lem-Alignment}-(iii), which are centered at $R_{j}$ with an arbitrary small radius $\epsilon_{j} >0$.

\end{remark}

\subsection{Main result}
The stability of the overall error dynamics of $\zeta_{3} \vcentcolon = \left[ \zeta^{T}_{1},\zeta^{T}_{2} \right]^{T}\in\mathbb{R}^{18}$ described by \eqref{eq:xTldPeta}-\eqref{eq:TefP} and \eqref{eq:bTildP}-\eqref{eq:zP2} is summarized in the following theorem. 

\begin{theorem}[\bf{Tracking control}]\label{thm3}
Consider the position controller $f=\|T\|$ with the thrust vector $T$ in \eqref{eq:CtrlPs} and the attitude controller \eqref{eq:AttCtrl} with the gyro-bias estimation \eqref{eq:bEst} in closed loop with the system \eqref{eq:KinX}-\eqref{eq:DynOm}. Choose the design parameters for
\begin{itemize}
    \item position controller: $K_{f} = K_{f}^{T}>0$, and $k>k_{x}>0$,
    \item attitude controller: $K_{c}=K^{T}_{c}>0$, $\lambda_{c}>0$, and $\alpha_1,\alpha_2>0$,
    \item gyro-bias observer: $\gamma_{f}>0$, $k_{i}>0$, and $\Lambda_{i}=\Lambda^{T}_{i}>0$, for $i=1,2,\ldots ,n$,
\end{itemize}
such that $0<k_x<g-\mu_d$, $k>k_x$, and $\lambda_{\min}(K_f)>\frac{k_x}{4m^2}$ 
and \eqref{eq:lo}-\eqref{eq:la} are satisfied. Then, the origin $\zeta_{3}(t)=0_{18\times 0}$ of the overall system  is practically stable.
\end{theorem}

\begin{proof}
Consider the following positive definite and radially unbounded function for $\zeta_{3}$
\begin{equation}\label{eq:V3}
    V_{3}(\zeta_{3},t) \vcentcolon =cV_{1}(\zeta_{1})+V_{2}(\zeta_{2},t),
\end{equation}
where $V_{1}$ and $V_{2}$ are given in \eqref{eq:V1} and \eqref{eq:V2}, respectively, and $c>0$ i a constant. Therefore, $\forall (\zeta_{1},\tilde{R},\tilde{\omega},\tilde{b})\in\mathcal{X}'_{e}$ 
\begin{equation}\label{eq:V3Bounds}
   \gamma_{1}\|\zeta_{3}\|^{2}\leq  V_{3}(\zeta_{3},t) \leq \gamma_{2}\|\zeta_{3}\|^{2},\quad   
\end{equation}
where
\begin{align}
   &\gamma_{1} \vcentcolon = \frac{1}{2}\min\left\{cm,ck^{2}_{x},c, \lambda_{\min}(M),1,\alpha_{2} \right\},\notag\\
   &\gamma_{2} \vcentcolon = \frac{1}{2} \max\left\{cm ,ck^{2}_{x}, c, \lambda_{\max}(M),1,\alpha_{2}+\beta \right\}.\notag
\end{align}
Taking the time derivative of \eqref{eq:V3}, by using \eqref{eq:V1p} and \eqref{eq:V2P}, it yields
\begin{align}
    \dot{V}_{3} =& -cm\left( mk -k_{x}\right)\eta^{T} \eta  - ck^{3}_{x}\tilde{x}^{T}\tilde{x}
    - cy^{T}K_{f}y -c\frac{k^{2}_{x}}{m} y^{T}\tilde{x}+ cf\eta^{T}(R-R_{d})e_{z}+c\eta^T f_d  \notag\\
    &- \lambda_{c}z^{T}\left( \alpha_{1}I_{3}+\alpha_{2}J\right)z-\tilde{\omega}^{T}K_{c}\tilde{\omega}   - \tilde{b}^{T}K_{f}\tilde{b} +f_d\tilde{\omega}^{T}G\tilde{b} +\tilde{\omega}^T\tau_d  \notag \\
     \leq& -cm\left( k - k_{x}\right)\|\eta\|^{2}  -ck^{3}_{x}\|\tilde{x}\|^{2}
     -c\lambda_{\min}(K_{f}) \|y\|^{2} + c\frac{k^{2}_{x}}{m}\|y\|\|\tilde{x}\|+ cf_{max}\|R-R_{d}\|\|\eta\| +c\|\eta\| \| f_d\|  \notag\\
    &- \lambda_{c}\lambda_{a}\|z\|^{2} -\lambda_{\min}(K_{c})\|\tilde{\omega}\|^{2}  - \lambda_{o}\|\tilde{b}\|^{2} + \|G\|\|\tilde{\omega}\|\|\tilde{b}\|+\|\tilde\omega\|\|\tau_d\| .
    \notag
\end{align}
In view of \eqref{eq:RRdT} of Lemma \ref{lem-Alignment}-(iv), it has 
\begin{align}
    \dot{V}_{3}&\leq  -\lambda_{\min}(Q_{3})\left( \|z\|^{2}+\|\eta\|^{2} \right) -\lambda_{\min}(Q_{4})\left(\|\tilde{\omega}\|^{2} + \|\tilde{b}\|^{2}\right)
    -\lambda_{\min}(Q_{5})\left( \|\tilde{x}\|^{2} + \|y\|^{2}\right)+c\|\eta\| \| f_d\|+\|\tilde\omega\|\|\tau_d\|,\notag
\end{align}
where 
\begin{align}
    Q_{3} = \left[ \begin{array}{cc}
         \lambda_{c}\lambda_{a} & -c\frac{f_{max}}{2}\sqrt{\frac{\varpi \beta}{\alpha_{1}}} \\
         -c\frac{f}{2}\sqrt{\frac{\varpi \beta}{\alpha_{1}}} &   cm\left( mk -k_{x}\right)
    \end{array}\right] , \ 
    Q_{4} = \left[ \begin{array}{cc}
         \lambda_{\min}(K_{c}) & -\frac{1}{2}\|G\| \\
         -\frac{1}{2}\|G\| &   \lambda_{o}
    \end{array}\right] ,  \
    Q_{5} =\left[ \begin{array}{cc}
         ck^{3}_{x} & -c\frac{k^{2}_{x}}{2m} \\
         -c\frac{k^{2}_{x}}{2m} & c\lambda_{\min}(K_{v}) 
    \end{array}\right]. \notag
\end{align}
Note that matrices $Q_{5}$ and $Q_{4}$ are positive definite under conditions  $\lambda_{\min}(K_f)>\frac{k_x}{4m^2}$ 
and \eqref{eq:lo}, and matrix $Q_{3}$ is positive definite under condition \eqref{eq:la} provided that $k>k_{x}$, for any constant $c$ that satisfies
\begin{equation*}
    0<c<\frac{\alpha_{1}\lambda_{c}\lambda_{a}\left( k-k_{x}\right)}{mg^{2}\varpi \beta},
\end{equation*}
given that $f< 2mg$ by condition \eqref{eq:TCond}. 
Thus,
\begin{align}\label{eq:V3p}
    \dot{V}_{3} &\leq -\gamma_{3} \|\zeta_{3}\|^{2}+ w_d\|\zeta_3\|, 
\end{align}
where $w_d:=c \| f_d\|+\|\tau_d\| $, and  $\gamma_{3}$ given by 
\begin{equation}\label{eq:gamma3}
    \gamma_{3} \vcentcolon = \min\left\{ \lambda_{\min}(Q_{3}),\lambda_{\min}(Q_{4}),\lambda_{\min}(Q_{5}) \right\}.
\end{equation}
Therefore, the origin $\zeta_3=0_{18\times 1}$ is practically stable by Lemma \ref{lem-pracStab}-(ii).

Since the domain of attraction $\mathcal{X}'_e$ can be arbitrarily large to cover almost the entire state space by properly choosing the design parameter such that $ SO(3)\backslash \mathcal{B}_{\epsilon}$ covers almost the entire group $SO(3)$, the almost global 
practical stability of  $(\zeta_{1},\tilde{R},\tilde{\omega},\tilde{b}) = (0_{9\times 1},I_{3},0_{3\times 1},0_{3\times 1})$ is concluded.

\end{proof}

\begin{remark}[The position-tracking controller]
    Compared to previous work, the proposed scheme considers the more challenging {\it position-tracking} task instead of {\it position regulation} \citep{tayebi2013inertial,roberts2013new}. In addition, vector measurements and {\it biased} (instead of bias-free) gyro-rate measurements are used directly without attitude reconstruction. When the task is position regulation, the proposed scheme does not need linear velocity for its implementation, similar to the controller of \cite{roberts2013new}. Furthermore, the force and torque disturbances are taken into account in the proposed scheme, and global practical stability 
    is proved. Moreover, the rotation matrix, instead of unit quaternions or other three-parameter representations of the attitude, is employed in the development of the controller, which allows  addressing attitude control without concern for problems related to singularities in the kinematics or double-covering of the unit quaternion group. 
\end{remark}

\begin{remark}[Almost global practical stability]\label{rmk-semiglobal}
In the absence of external disturbances, the almost global exponential stability of the desired equilibrium is achieved under the position controller $f=\|T\|$ with the thrust vector $T$ in \eqref{eq:CtrlPs} and the attitude controller \eqref{eq:AttCtrl} in Theorem \ref{thm3} compared with local exponential stability in \cite{bertrand2011hierarchical}. Note that the proposed position controller and attitude controller can be designed independently, provided (1) the thrust force $f$ in the position controller is bounded, and (2) the difference $\|R-R_d\|$ between the actual attitude $R$ and the desired attitude $R_d$ in the attitude controller goes to zero. Both conditions are ensured by Proposition \ref{thm2},  Theorem \ref{thm3}, and Lemma \ref{lem-Alignment}-(iv). This modular design brings several salient features, such as stability analysis and control parameter tuning. 
\end{remark}


\begin{remark}[Practical issues]\label{rmk6}
To implement tracking controllers, the following practical issues can be taken into account.
\begin{itemize}
    \item {\it Desired angular velocity}. The attitude controller \eqref{eq:AttCtrl} depends on the desired angular acceleration $\dot{\omega}_{d}$ through \eqref{eq:omgREstp}, which in turn requires the second time derivative of the thrust control \eqref{eq:CtrlPs} given that $\dot{\omega}_{d}(T,\dot{T},\ddot{T},\psi_{d},\dot{\psi}_{d},\ddot{\psi}_{d}) =  \left(R^{T}_{d}\ddot{R}_{d} + \dot{R}^{T}_{d}\dot{R}_{d}\right)^{\vee}$, thus the acceleration error $\dot{\tilde{v}}$ is required for the implementation. To avoid $\dot{v}$, a filtered version of $\omega_{d}$ to approximate $\dot{\omega}_{d}$ as proposed \cite{zuo2014adaptive} may be used as in the following 
\begin{align}
    \dot{\vartheta}_{1} &= \vartheta_{2},\notag \\
    \dot{\vartheta}_{2} &= -2A\vartheta_{2} -A^{2}(\vartheta_{1}-\omega_{d}), \label{eq:2ndFlt}
\end{align}
where $A\in\mathbb{R}^{3\times 3}$ is a diagonal positive definite matrix. Then $\vartheta_{2}\in\mathbb{R}^{3}$ is an approximation of $\dot{\omega}_{d}$. The boundedness of $\|\vartheta_{2}-\dot{\omega}_{d}\|$ is ensured by \cite{zuo2014adaptive}, which contributes to a bounded disturbance $\tau_{d}=M(\vartheta_{2}-\dot{\omega_{d}})$ when using $\vartheta_{2}$ instead of $\dot{\omega}_{d}$ in \eqref{eq:omgREstp}.  The stability of the overall system is ensured by the robustness property in Proposition \ref{thm3}.

\item {Linear velocity measurement.}
The position controller \eqref{eq:CtrlPs} can be implemented without linear velocity measurement $v$ as follows. Recall that $y:=\mathrm{Tanh}(e_{f})$, then $\dot y$ in \eqref{eq:TefP} can be obtained similar to that of \cite{zhang2000global} as 
\begin{align}
    y &= p - k\tilde{x} ,\label{eq:h}\\
    \dot{p} &= -k\left( k_{x}\tilde{x}+p-k\tilde{x} \right) -K_{f}(p-k\tilde{x})  + k^{2}_{x}\left(1-\frac{1}{m}\right)\tilde{x} ,\quad
    p(0) = k\tilde{x}(0). \label{eq:pp}
\end{align}
Thus $y$ can be calculated using only position error $\tilde{x}$.  This feature is useful for non-aggressive applications when a fast inner attitude control loop ensures $R\approx R_d$  \citep{serra2016landing} or, for position regulation \citep{bertrand2011hierarchical} since, in this case, no desired angular velocity and thus no linear velocity is needed in the controller implementation.

\item {\it Vector measurements.} 

There are several ways to acquire vector measurements required for the attitude controller, such as using two magnetometers placed in the UAV with different orientations, a CCD camera or an IMU with one magnetometer and one accelerometer. In this last case, the measured body acceleration $v_a=R^Tr_a$, with $r_a$ depending on the apparent acceleration $ge_z+\dot v$ 
which is not a known constant inertial reference vector. However, when $\|\dot v\|<< g$, $r_a \approx e_z +\bar r_a$, where $\| \bar r_a\|<<1$, which can be considered as a bounded disturbance and tolerated by the robustness of the controller. 
\end{itemize}

\end{remark}

\section{Simulations}\label{Sec:Sim}
To illustrate the main result stated in Theorem \ref{thm3} numerical simulations were carried out. Four scenarios were considered. In Scenario 1, the vector, gyro rate, position, and linear velocity measurements are noise free, and the inertia matrix is known; In Scenario 2 the position, linear velocity, inertial vector, and angular velocity measurements are contaminated by noise, while the inertia matrix remains known; In Scenario 3, a parametric uncertainty of $\pm 30\%$ is added in the inertia matrix while the measurements are noise free; in Scenario 4, the apparent acceleration is measured; therefore, the corresponding inertial reference is not constant, as commented in Remark \ref{rmk6}. In all scenarios, the quadrotor parameters were taken from the experimental UAV from \cite{guerrero2011bounded}, the initial conditions and the design parameters were the same and are shown in Table \ref{tab:Parameters}. The desired trajectory was chosen as a Lemniscata calculated by $x_{d}(t)=[2.5\cos((2\pi/60)t),3\sin((2\pi/60)t)\cos((2\pi/60)t),3]^{T}$ (m). The desired angular acceleration $\dot{\omega}_{d}$ was approximated by the linear filter \eqref{eq:2ndFlt}.

\begin{table}[h]
 \caption{Parameters and initial conditions, for $i=1,2,3$. \label{tab:Parameters}}
 \centering
 \begin{tabular}{@{}lll@{}}
 \hline
 Initial Condition & Value & Unit\\
 \hline
 $\omega(0)$ & $0_{3\times 1}$ & (rad/s)\\
 $R(0)$ & $I_{3}$ & \\
 $x(0)$ & $0_{3\times 1}$ & (m)\\
 $v(0)$ & $0_{3\times 1}$ & (m/s) \\
 $\psi_{d}(0)$ & $\pi /4$ & (rad) \\
 $\dot{\psi}_{d}(0)$& $0$ & (rad/s)\\
 $e_{f}(0)$ &$0_{3\times 1}$ & \\
 $\bar{b}(0)$ &$0_{3\times 1}$& \\
 \hline \hline 
 Parameter & Value & Unit \\\hline
 $M$ & $\mathrm{diag}\{8.28,8.28,15.7\}\times 10^{-3}$& Kgm$^{2}$\\
 $m$ & $0.467$ & Kg\\
 $b$ & $[0.2,0.1,-0.1]^{T}$ & (rad/s) \\
 \hline
 Inertial Vectors & & \\
 \hline
 $r_{1}$ & $[0,0,1]^{T}$ & \\
 $r_{2}$ & $(1/\sqrt{3})[1,1,1]^{T}$ & \\
 $r_{3}$ & $r_{1}^\wedge r_{2}/ \|r_{1}^\wedge r_{2}\|$ & \\
 $k_{i}$ & $0.1$ & \\
 \hline
 Position Controller \eqref{eq:CtrlPs} & & \\
 \hline
 $k$ & $4$ & \\
 $k_{x}$ & $0.1$ & \\
 $K_{f}$ & $I_{3}$ & \\
 \hline
 Attitude Controller \eqref{eq:AttCtrl}& & \\
 \hline
 $K_{c}$ & $I_{3}$ & \\
 $\lambda_{c}$ & $1$ & \\
 $\alpha_{1}$ & $1$ & \\
 $\alpha_{2}$ & $0.01$ & \\
 \hline
 Bias Observer \eqref{eq:bEst}& & \\
 \hline
 $\Lambda_{i}$ & $10 I_{3}$ & \\
 $\gamma_{f}$ & $10000$ & \\
 \hline
 Filter \eqref{eq:2ndFlt}& & \\
 \hline
 $A$ & $20$ & \\
 \hline
 \end{tabular}
 \end{table}

\graphicspath{ {./Figs/} } 
\subsection{Scenario 1. Ideal situation.}

Figure \ref{fig:IdCase} shows the performance of the proposed controller in scenario 1, where the attitude error $\frac{1}{2}\mathrm{tr}(I_{3}-\tilde{R})$, the vector alignment error $z$, the angular velocity error $\tilde{\omega}$, the gyro bias estimation error $\tilde{b}$, and the position error $\tilde{x}$ are shown. The composite error $\eta$ and the auxiliary variable $\mathrm{Tanh}(e_{f})$ in Figs. \ref{fig:IdCase}(f)-(g). Note that after about $5 [s]$ transient, they converge to zero as established in Theorem \ref{thm3}. The control thrust force $f$ and the control torque $\tau$ are shown in Figs. \ref{fig:IdCase}(h) and (i), respectively. The thrust force is observed to reach a steady state value of $4.6$ (N) maintaining away from zero in all time, while the norm of torque control remains below $0.01$ (Nm).

\begin{figure}[h]
	\centering
		\includegraphics[trim = 17mm 0mm 17mm 0mm,clip,scale=0.26]{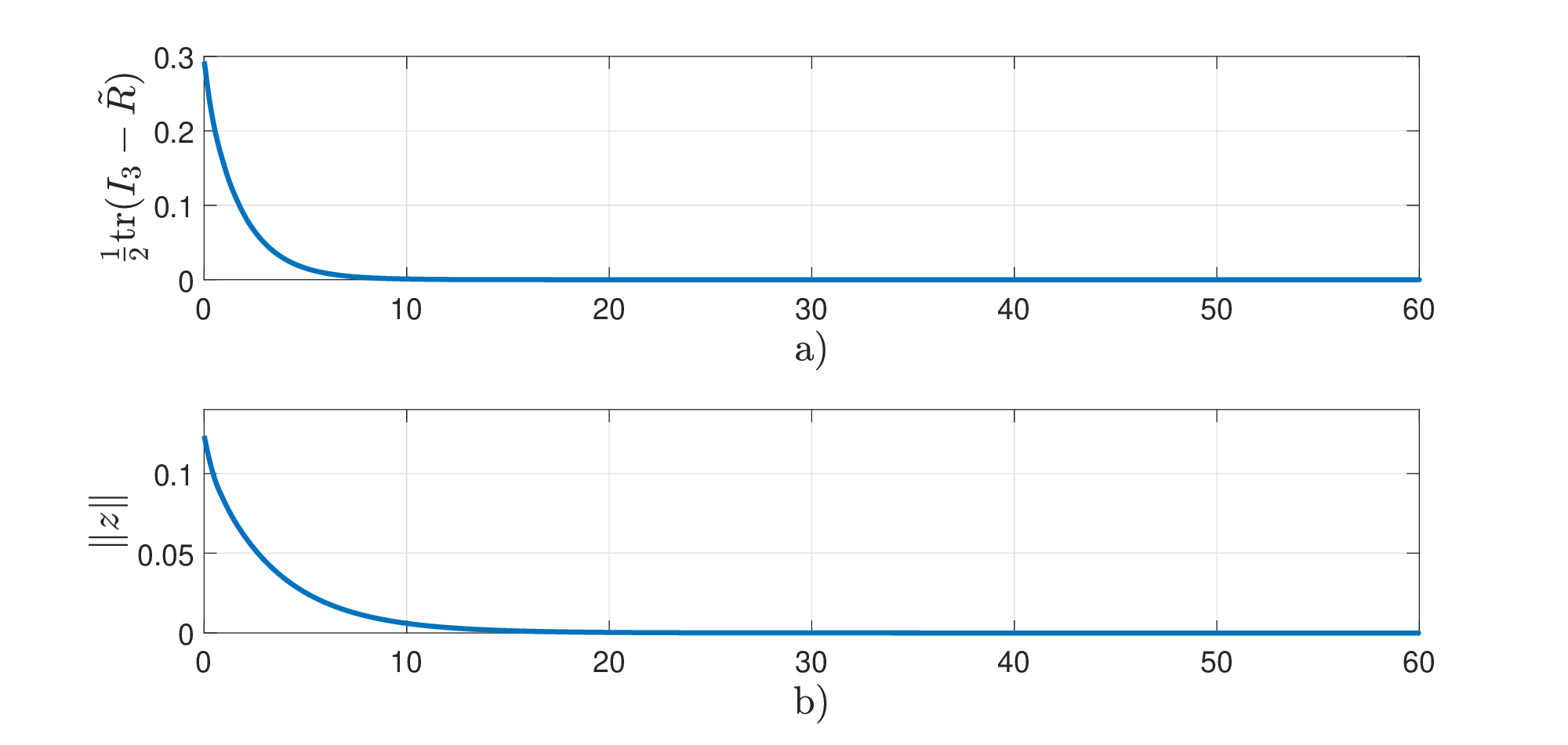}
		\includegraphics[trim = 17mm 0mm 17mm 0mm,clip,scale=0.26]{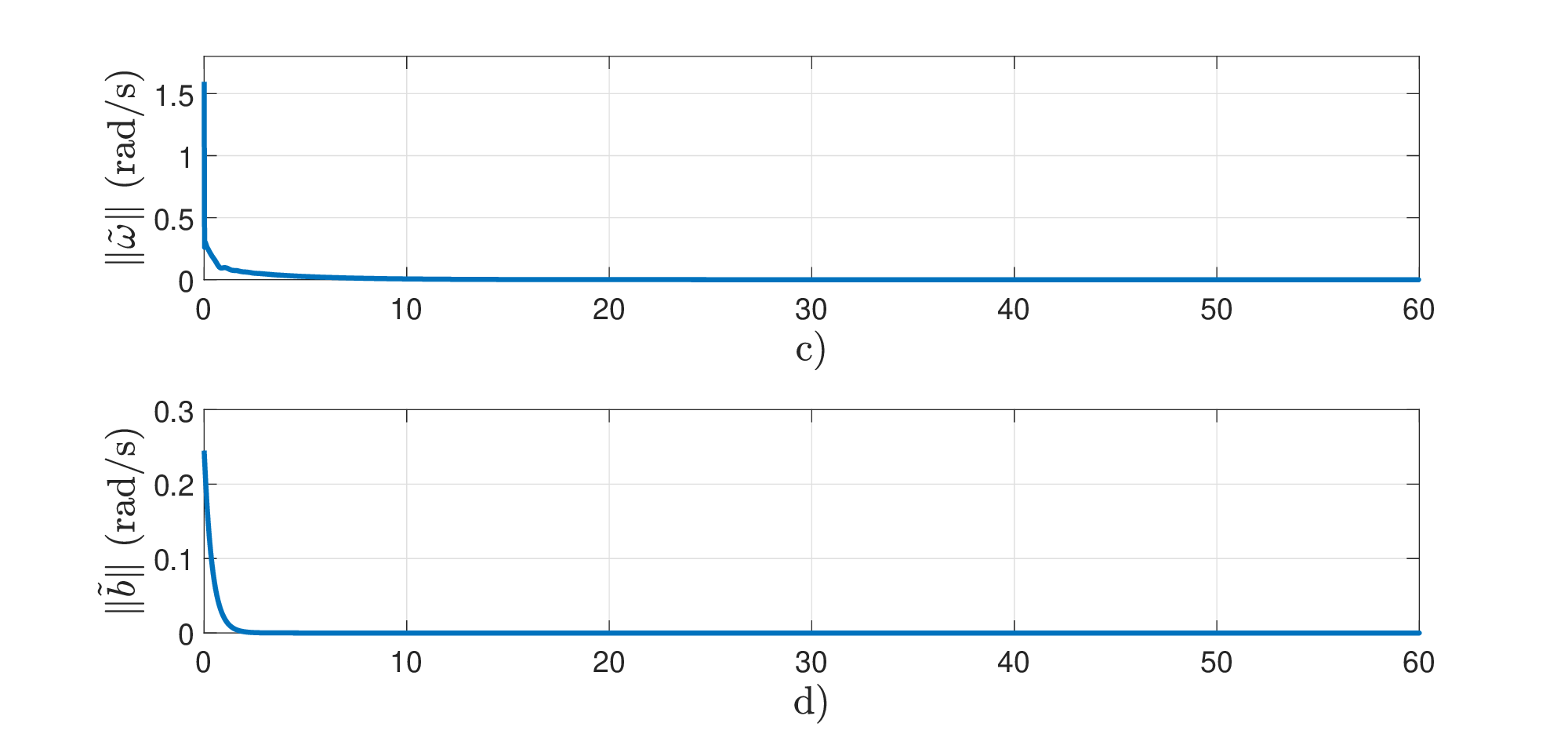}
		\includegraphics[trim = 17mm 0mm 17mm 0mm,clip,scale=0.26]{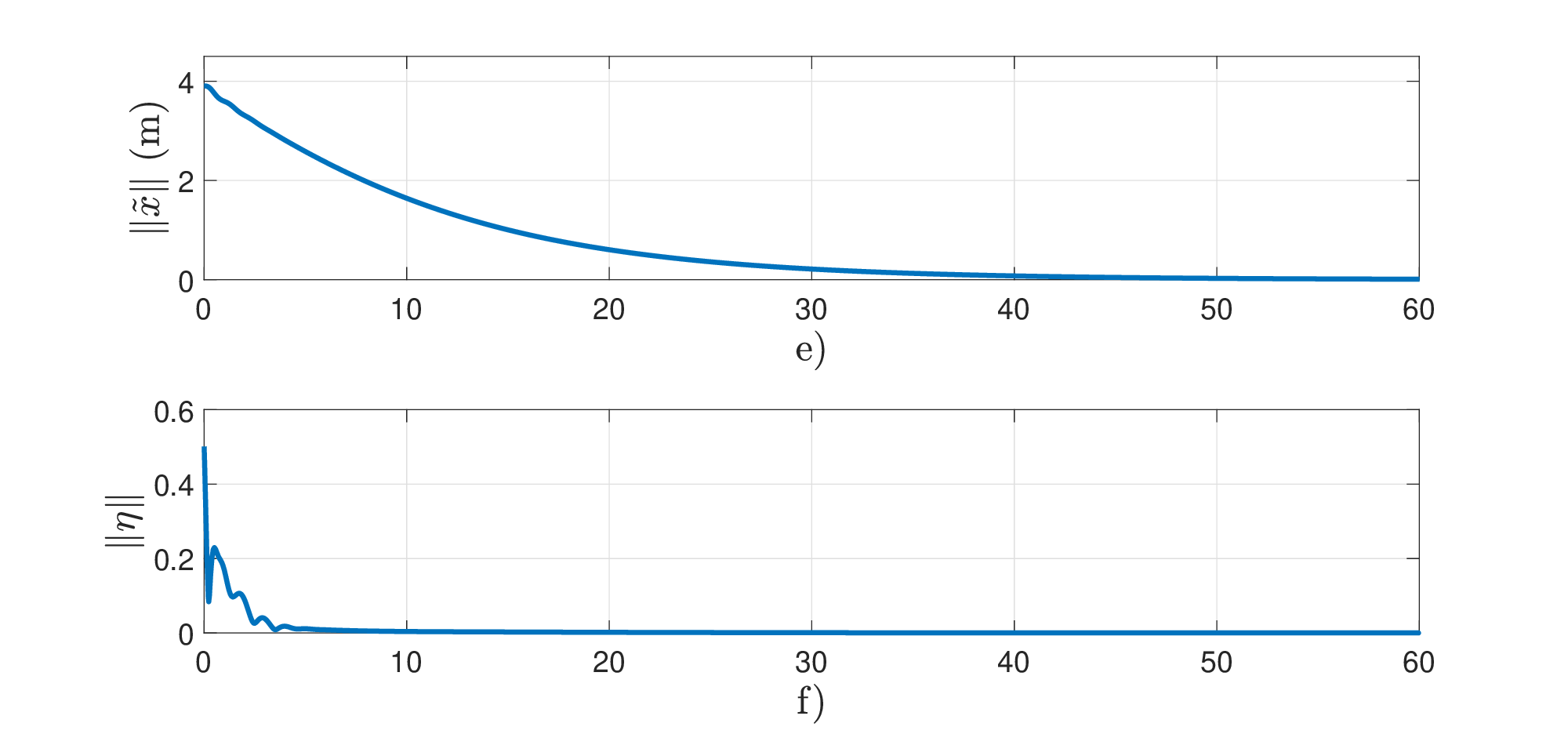}
		\includegraphics[trim = 17mm 0mm 17mm 0mm,clip,scale=0.26]{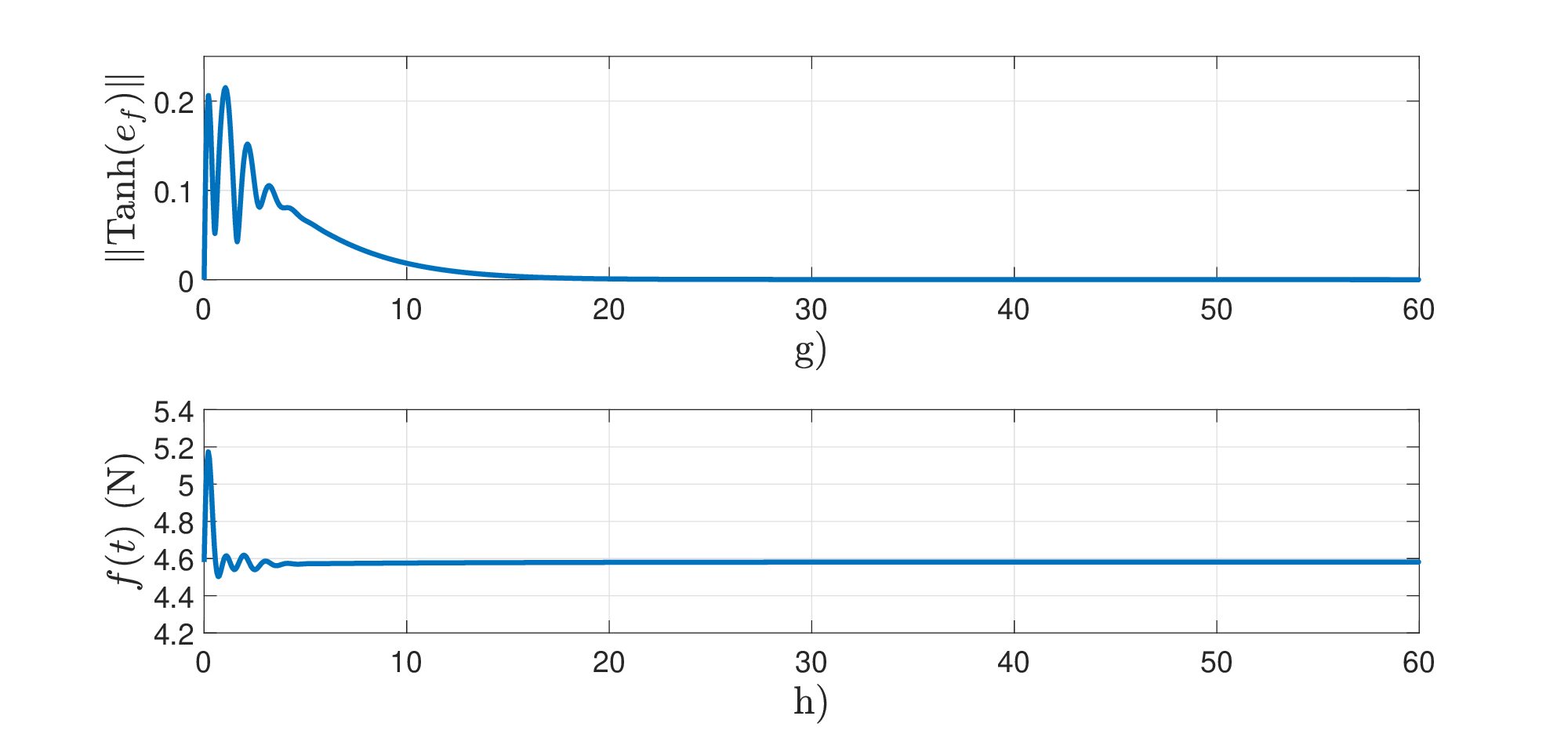}
		\includegraphics[trim = 17mm 80mm 17mm 0mm,clip,scale=0.26]{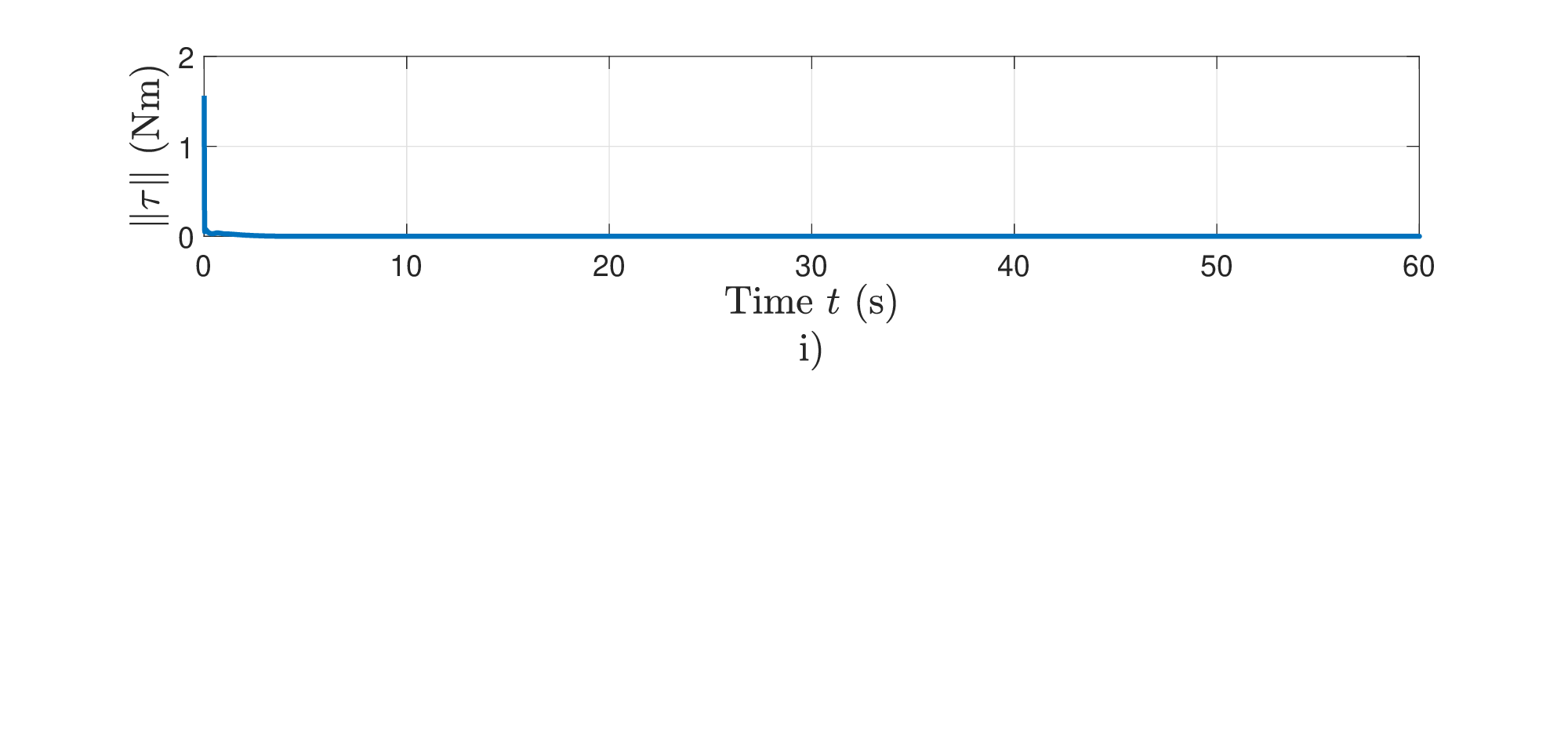}
		\caption{Scenario 1 (ideal situation). Performance of the proposed controller under noise-free measurements and known inertia matrix.}
		\label{fig:IdCase}
\end{figure}
  
A 3D plot is given in Fig. \ref{fig:DesTy} showing the actual path of the AUV and the desired one.  
\begin{figure*}[h]
  \centering
  \subfloat[]{\includegraphics[scale=0.26]{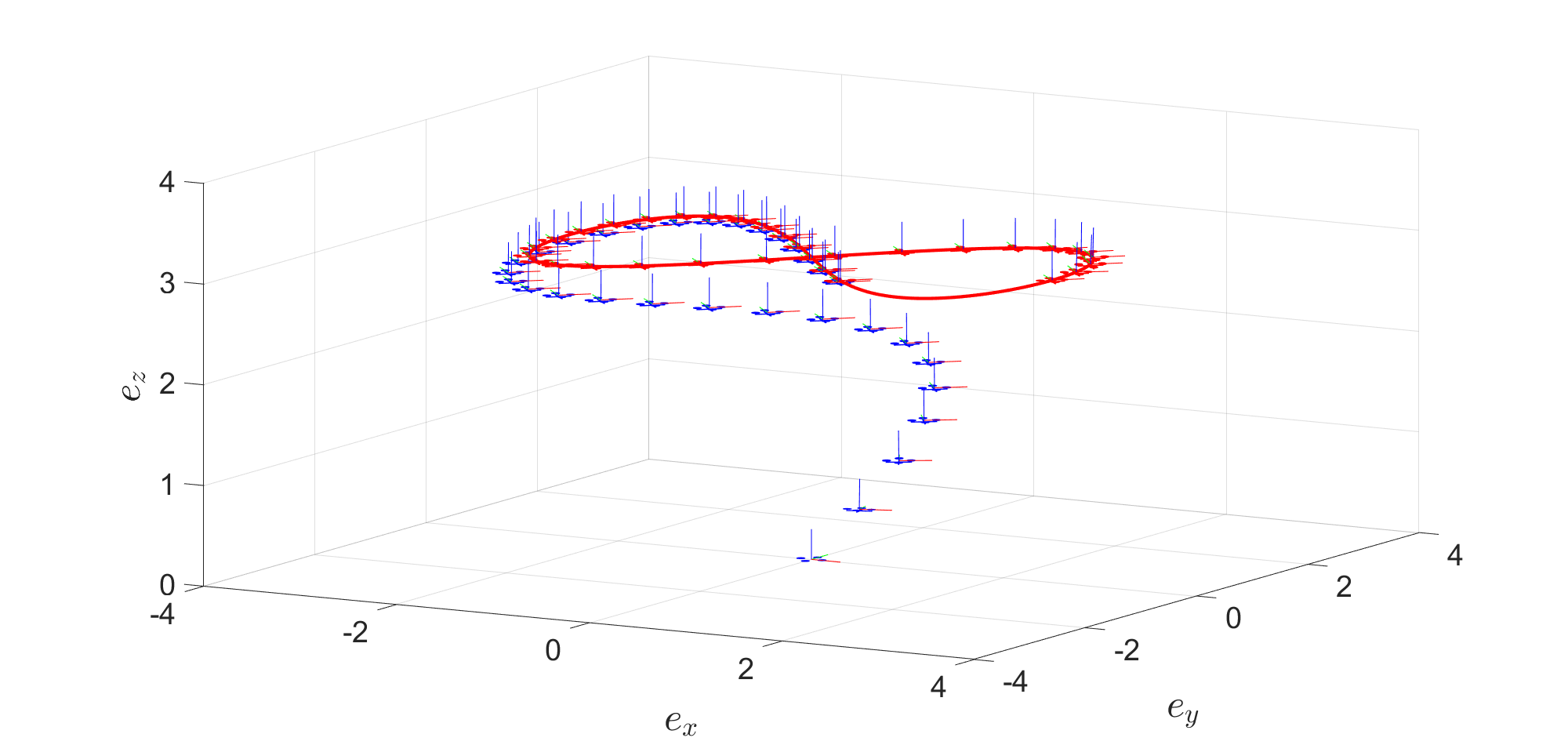}%
  \label{fig_first_case}}
  \hfil
  \subfloat[]{\includegraphics[scale=0.26]{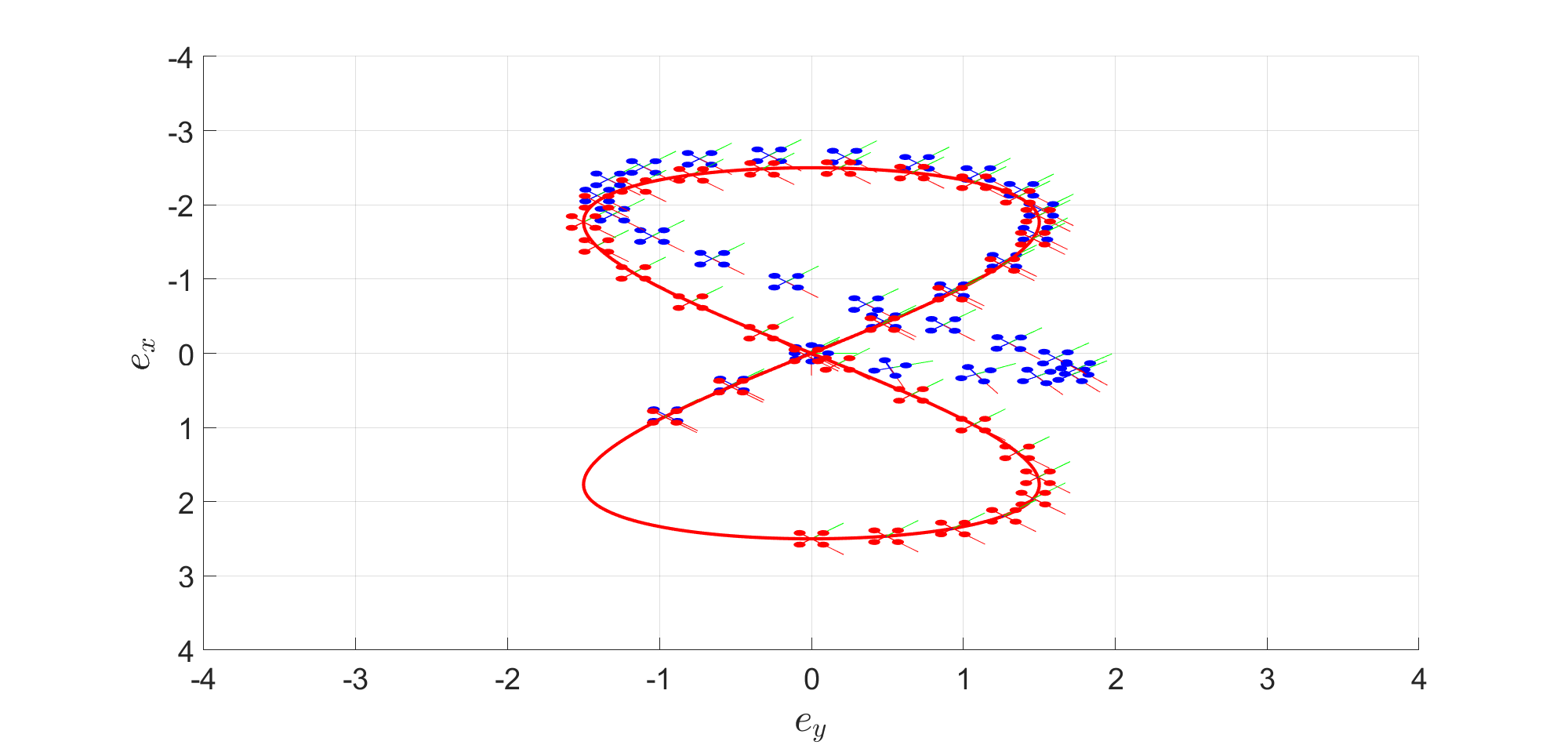}%
  \label{fig_second_case}}
  \vfil
  \subfloat[]{\includegraphics[scale=0.26]{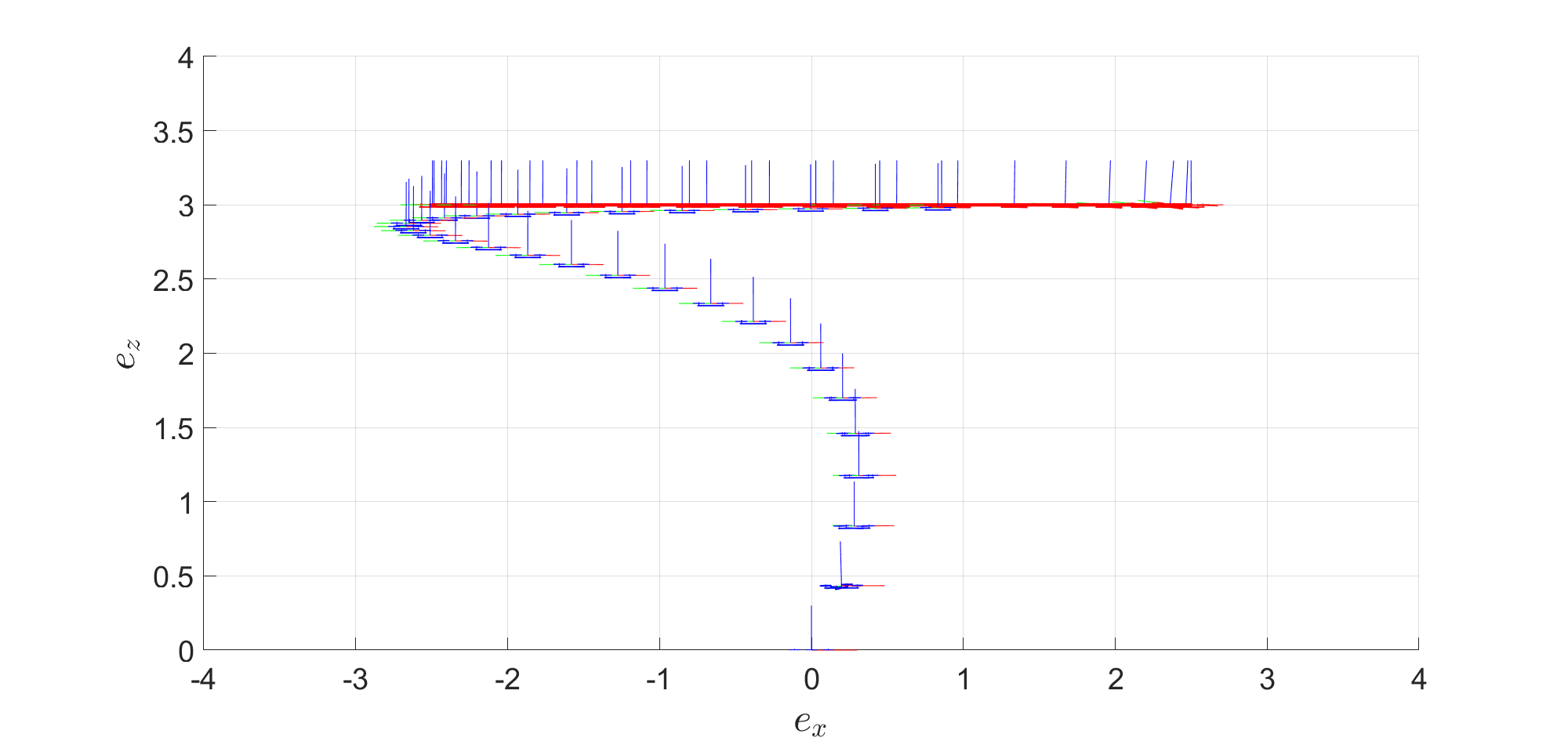}%
  \label{fig_third_case}}
  \caption{Scenario 1 (ideal situation). The actual path of the quadrotor (blue) tracks the desired path (red). (a) A 3D view in the inertial frame  $\{e_{x},e_{y},e_{z}\}$ [m]. (b) Plane $\{e_{x}-e_{y}\}$ view. (c) Plane $\{e_{x}-e_{z}\}$ view.}
  \label{fig:DesTy}
  \end{figure*}

\subsection{Scenario 2. Noisy measurements}

The noisy measurements of $x$, $v$, $\omega_{g}$, and $v_{i}$, for $i=1,2,3$, were simulated as follows: $x_{m}=x+\nu_{1}\varrho_{1}$, $v_{m}=v+\nu_{2}\varrho_{2}$, $\omega_{m}=\omega_{g}+\nu_{3}\varrho_{3}$, and $v_{m,i}= (v_{i}+\nu_{4}\varrho_{4})/\|v_{i}+\nu_{4}\varrho_{4}\|$, where $\varrho_{j}\in\mathbb{R}^{3}$ are zero-mean Gaussian distributions with unit variance for all $j=1,2,3,4$, and $\nu_{1,2}\in N(0,0.05)$, $\nu_{3,4}\in N(0,0.1)$  are uniform distributions. 

The attitude error in Fig. \ref{fig:NysCase}(a) shows no substantial changes compared to Scenario 1, while the norm of the alignment variable $z$ remained oscillating inside of $0.025$ (Fig. \ref{fig:NysCase}(b)), which is of the same magnitude of the noise density in the vector measurements $v_{m,i}$ given in Table \ref{tab:Parameters}). Note that the angular velocity error $\tilde{\omega} = \omega - \omega_{r}= \omega +\lambda_{c}z -\omega_{d} $ in Fig. \ref{fig:NysCase}(c) is most affected by noise due to the sum effect of noise in vector measurements $v_{m,i}$ through $z$, position $x_{m}$, linear velocity $v_{m}$ through $\omega_{d}(T,\dot{T},\psi_{d},\dot{\psi}_{d})$ in \eqref{eq:Rd}. The norm of the gyro-bias estimation error $\tilde{b}$ in Fig. \ref{fig:NysCase}(d) remains below $0.2$, which is due to the sum of the noise in the vector measurements $v_{m,i}$ and the gyro sensor $\omega_{m}$. Noise in $\tilde{b}$ can be reduced by decreasing the observer gain $\Lambda_{i}$ or increasing the filter gain $\gamma_{f}$,  without affecting the performance of the overall system. Figs. \ref{fig:NysCase}(e)-(f) illustrate the norm of the state variables $\tilde{x}$, $\eta$, and $\mathrm{Tanh}(e_{f})$, all remaining below $0.1$, according to the sum of the noise density in the position and speed measurements. The control thrust force $f$ and the torque control $\tau$ in  Figs. \ref{fig:NysCase}(h) and (i)  remained bounded and close to the values in Scenario 1, as established in Proposition \ref{thm3}. 

\begin{figure}[h]
	\centering
		\includegraphics[trim = 17mm 0mm 17mm 0mm,clip,scale=0.26]{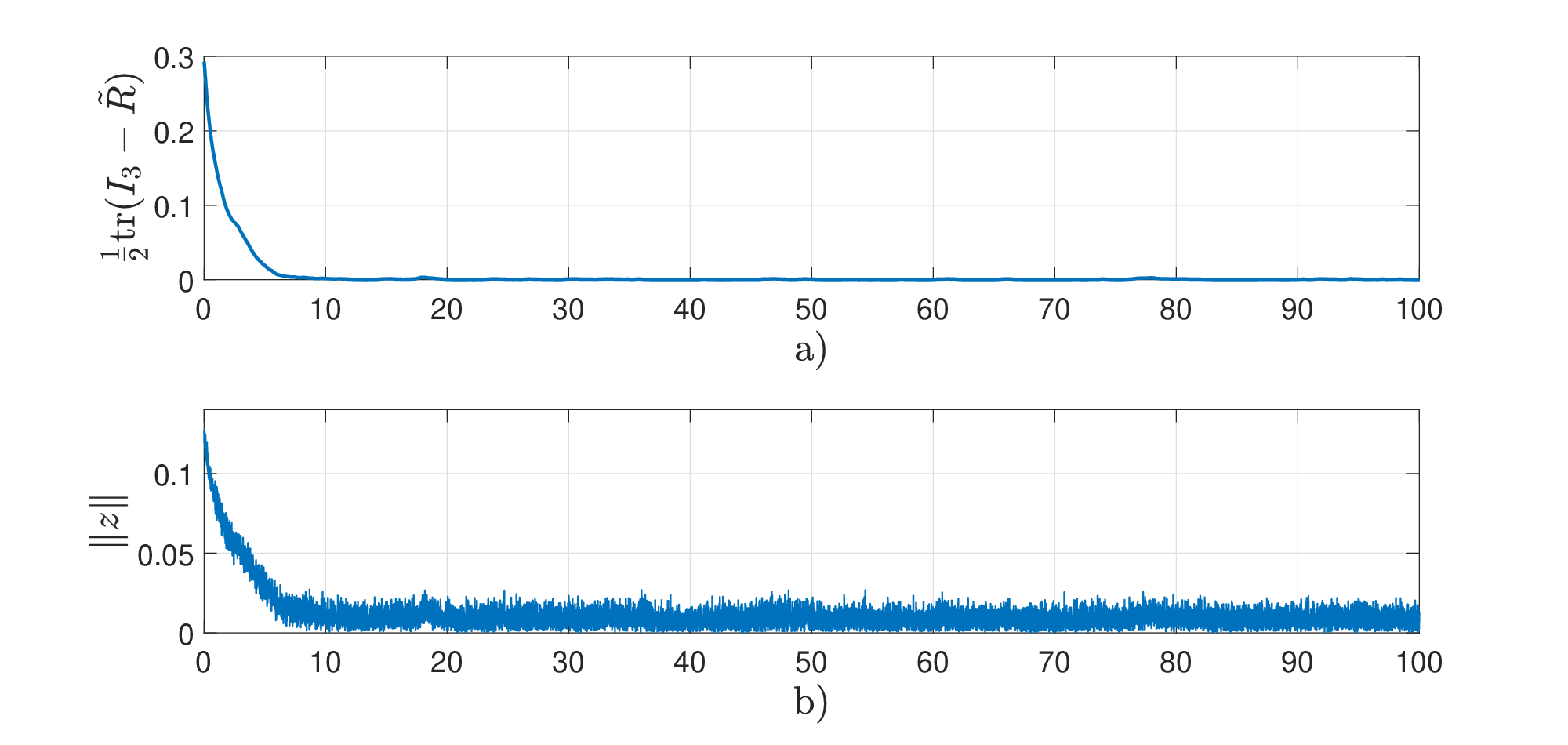}
		\includegraphics[trim = 17mm 0mm 17mm 0mm,clip,scale=0.26]{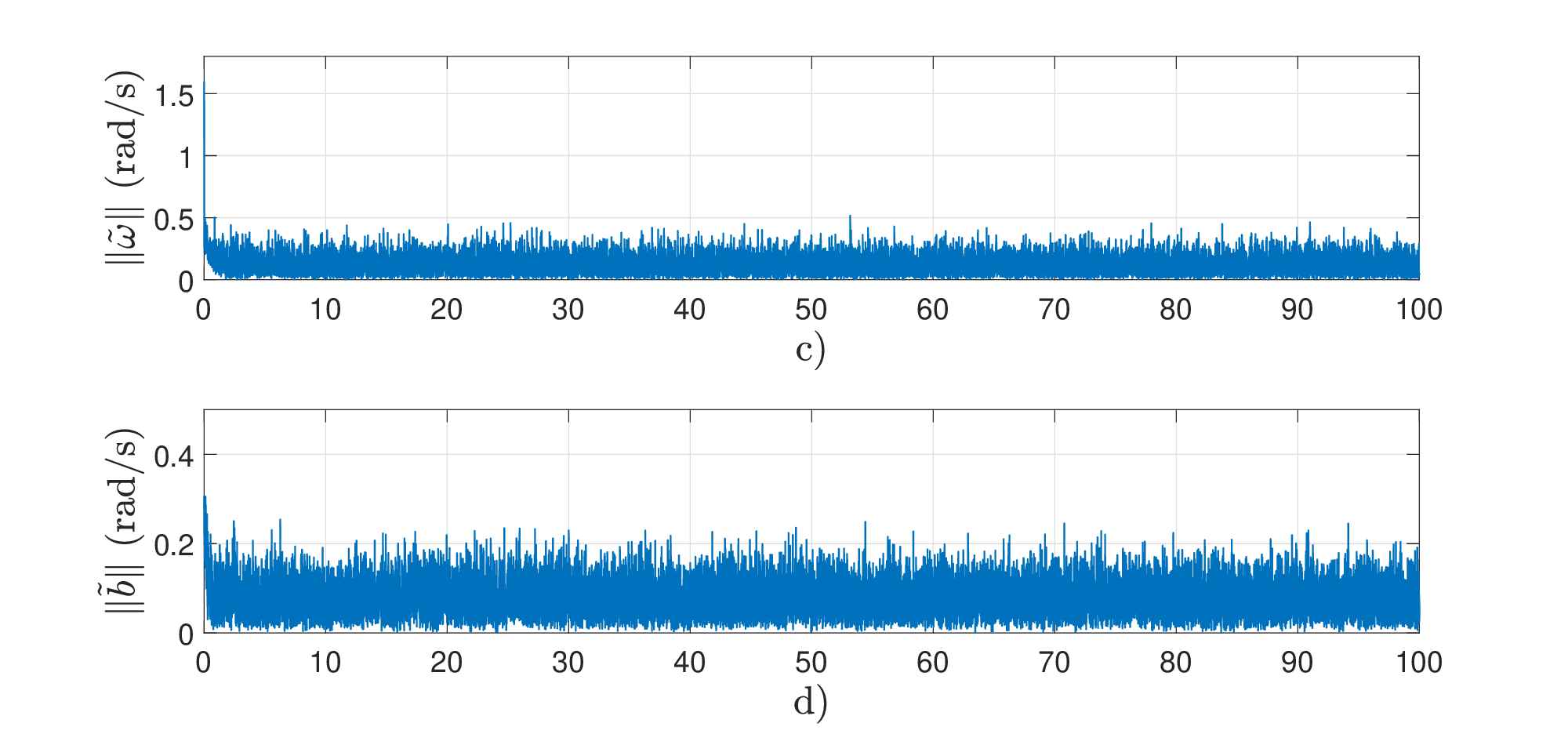}
		\includegraphics[trim = 17mm 0mm 17mm 0mm,clip,scale=0.26]{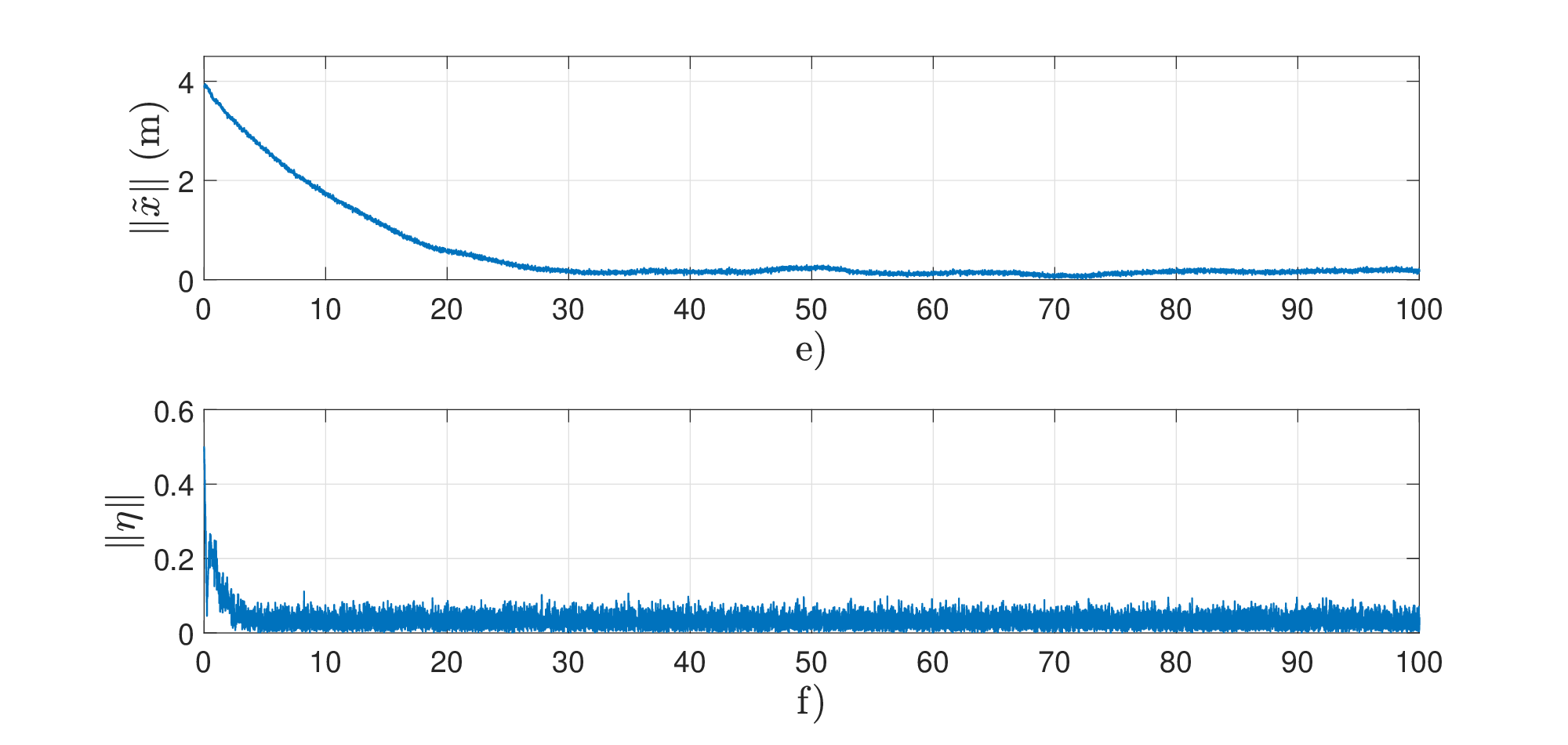}
		\includegraphics[trim = 17mm 0mm 17mm 0mm,clip,scale=0.26]{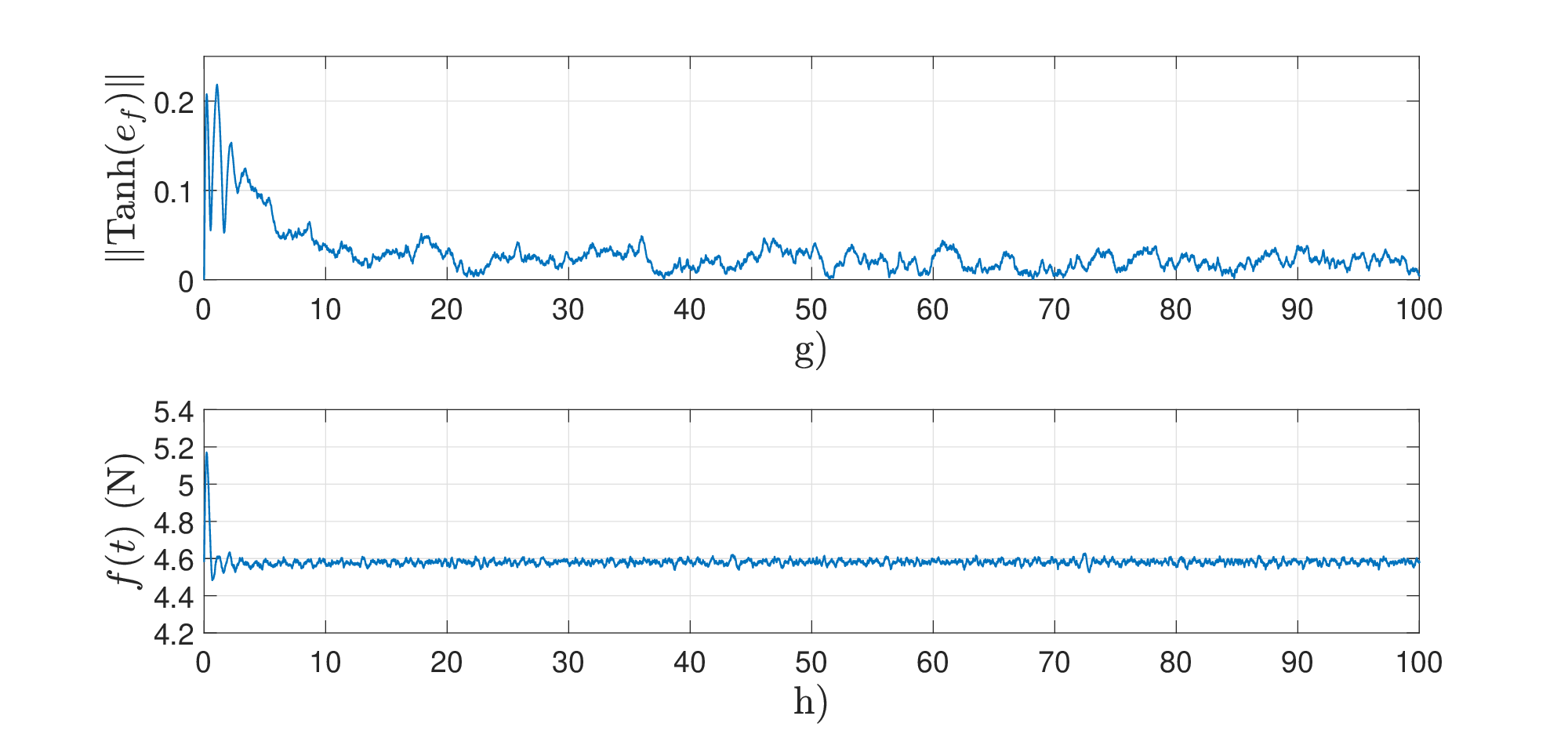}
		\includegraphics[trim = 17mm 80mm 17mm 0mm,clip,scale=0.26]{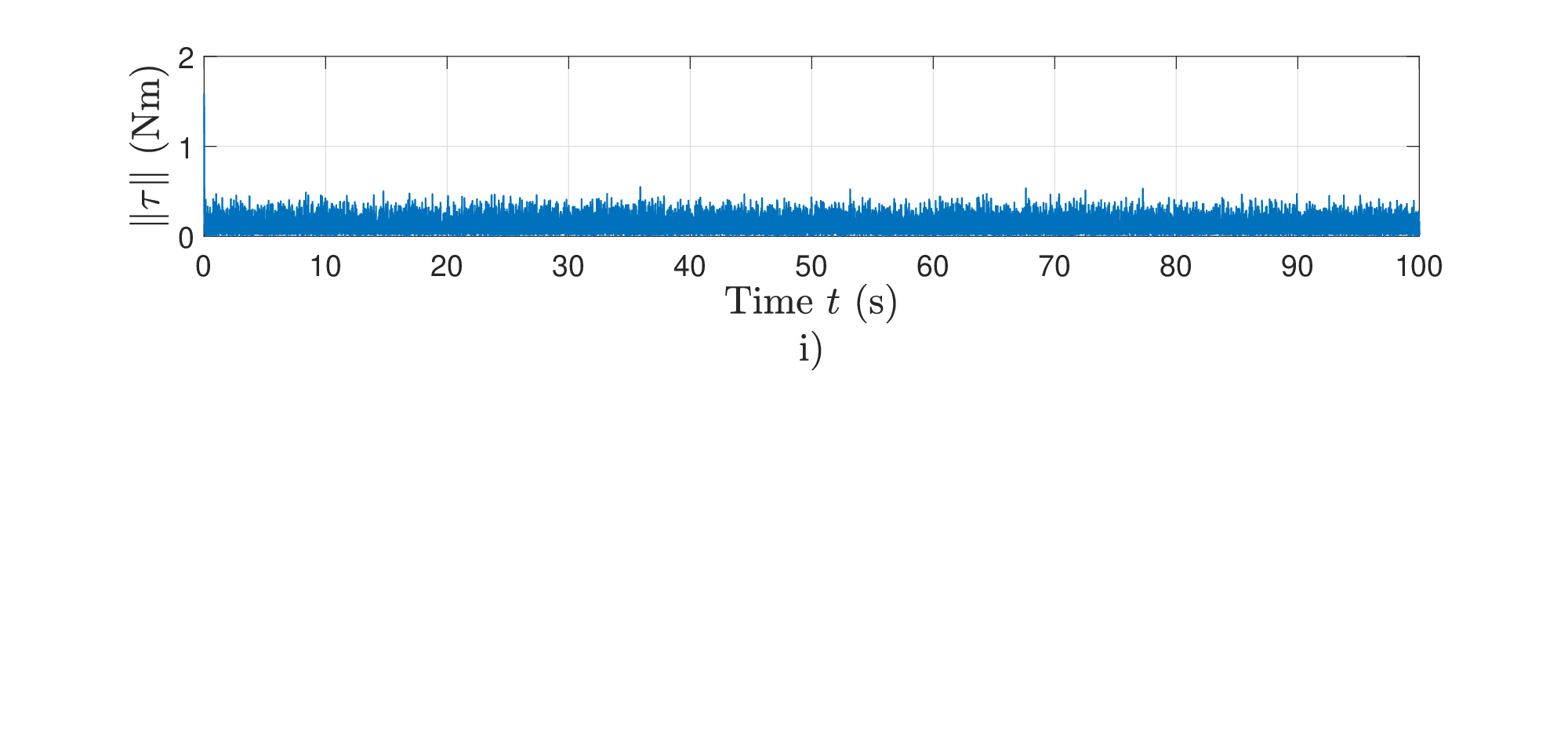}
		\caption{Scenario 2 (noisy-measurement situation). Performance of the proposed controller under noisy measurements.}
		\label{fig:NysCase}
\end{figure}

\subsection{Scenario 3. Parametric uncertainty.}

For this scenario, $\pm 30\%$ uncertainty in the inertia matrix with noise-free measurements is considered. Fig. \ref{fig:UPCase} shows the simulation results. No significant changes in performance are observed, with the convergence times of the error state and the control efforts practically the same. 

\begin{figure}[h]
	\centering
		\includegraphics[trim = 17mm 0mm 17mm 0mm,clip,scale=0.26]{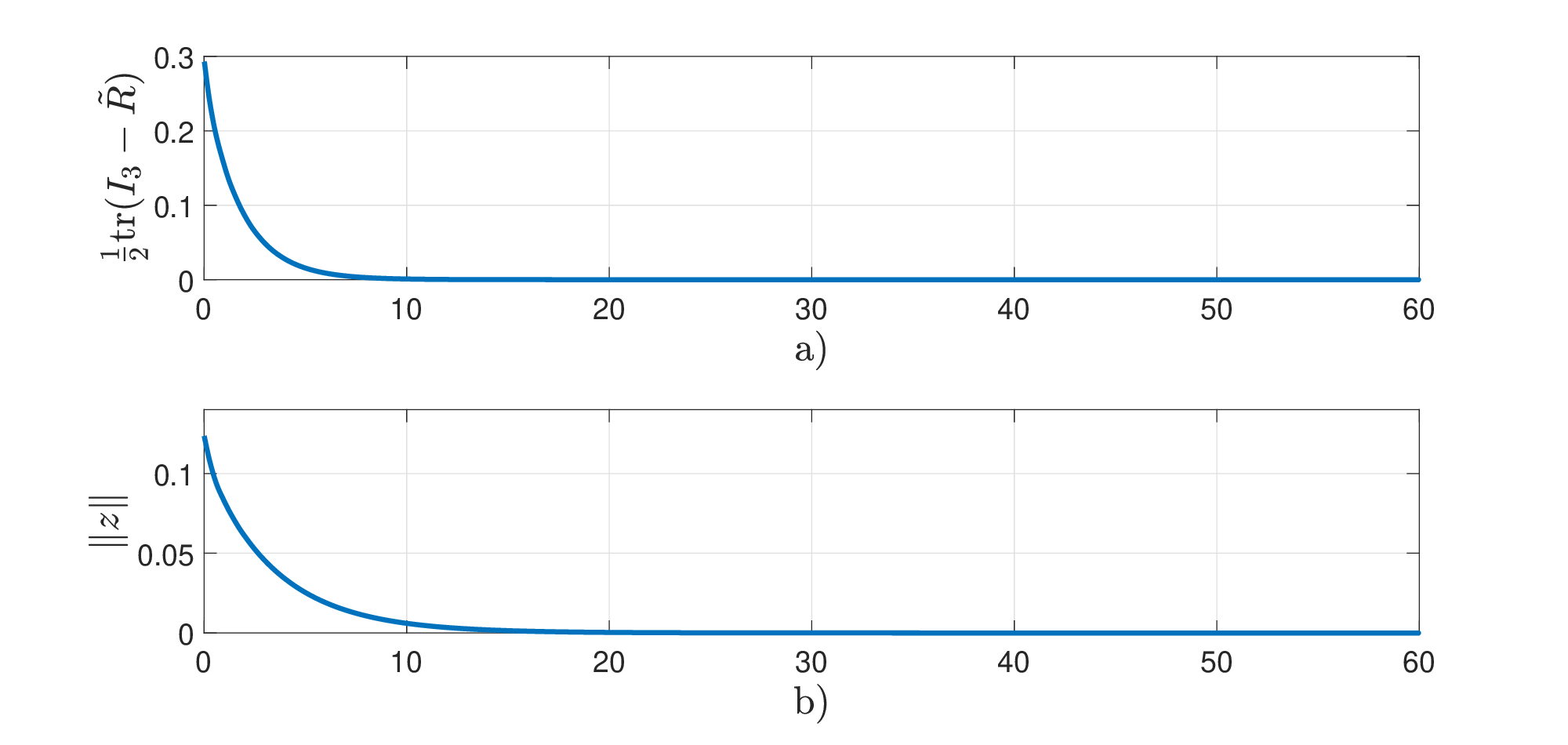}
		\includegraphics[trim = 17mm 0mm 17mm 0mm,clip,scale=0.26]{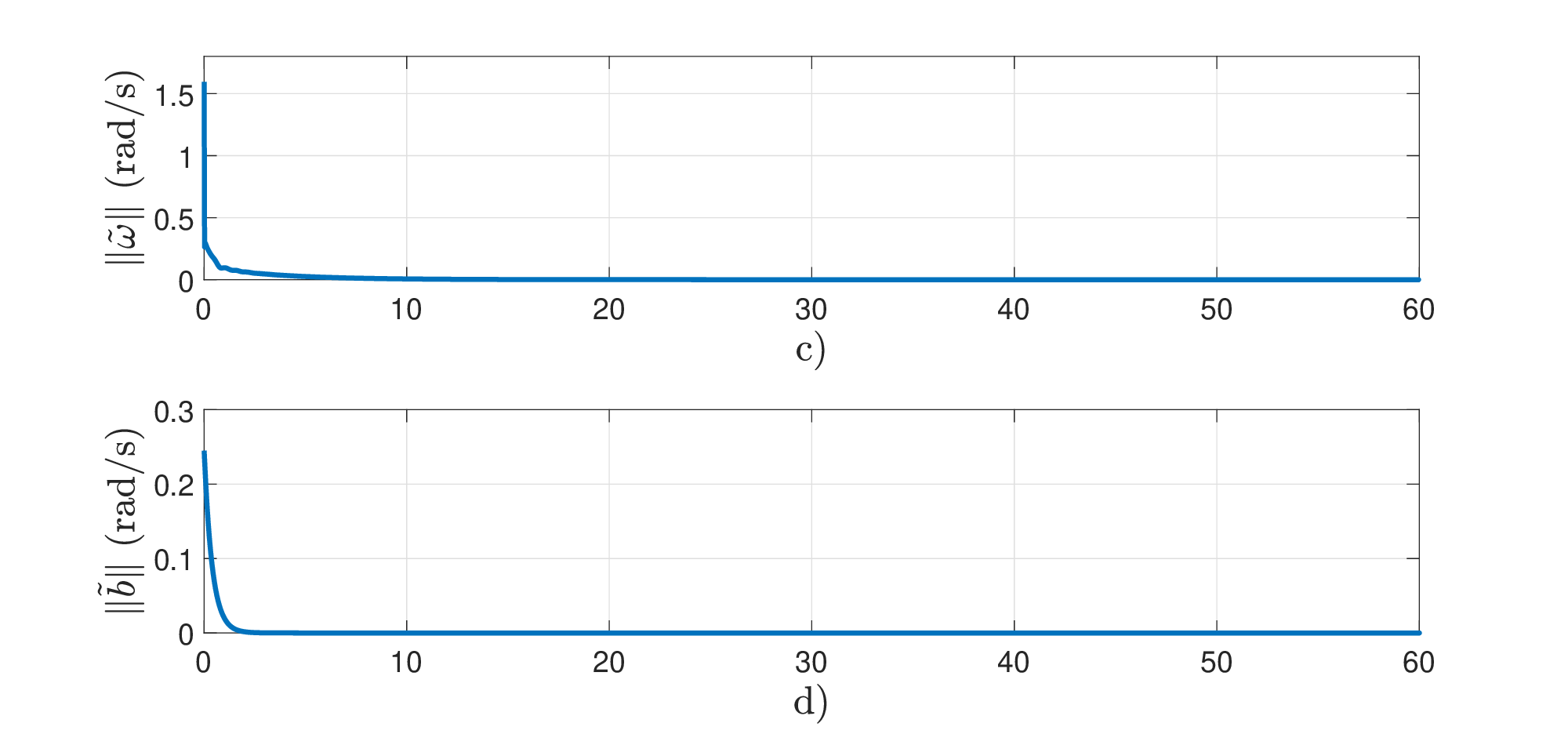}
		\includegraphics[trim = 17mm 0mm 17mm 0mm,clip,scale=0.26]{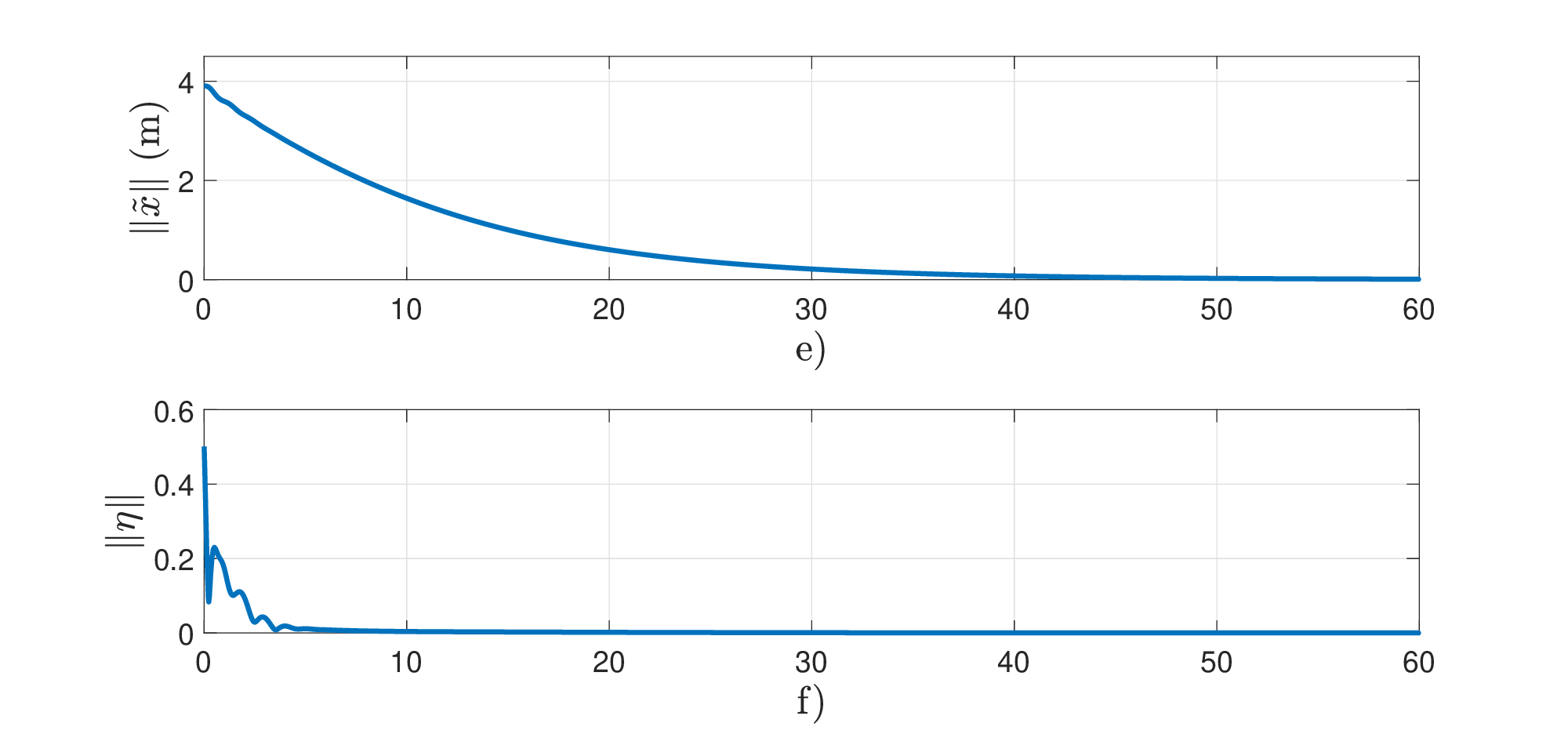}
		\includegraphics[trim = 17mm 0mm 17mm 0mm,clip,scale=0.26]{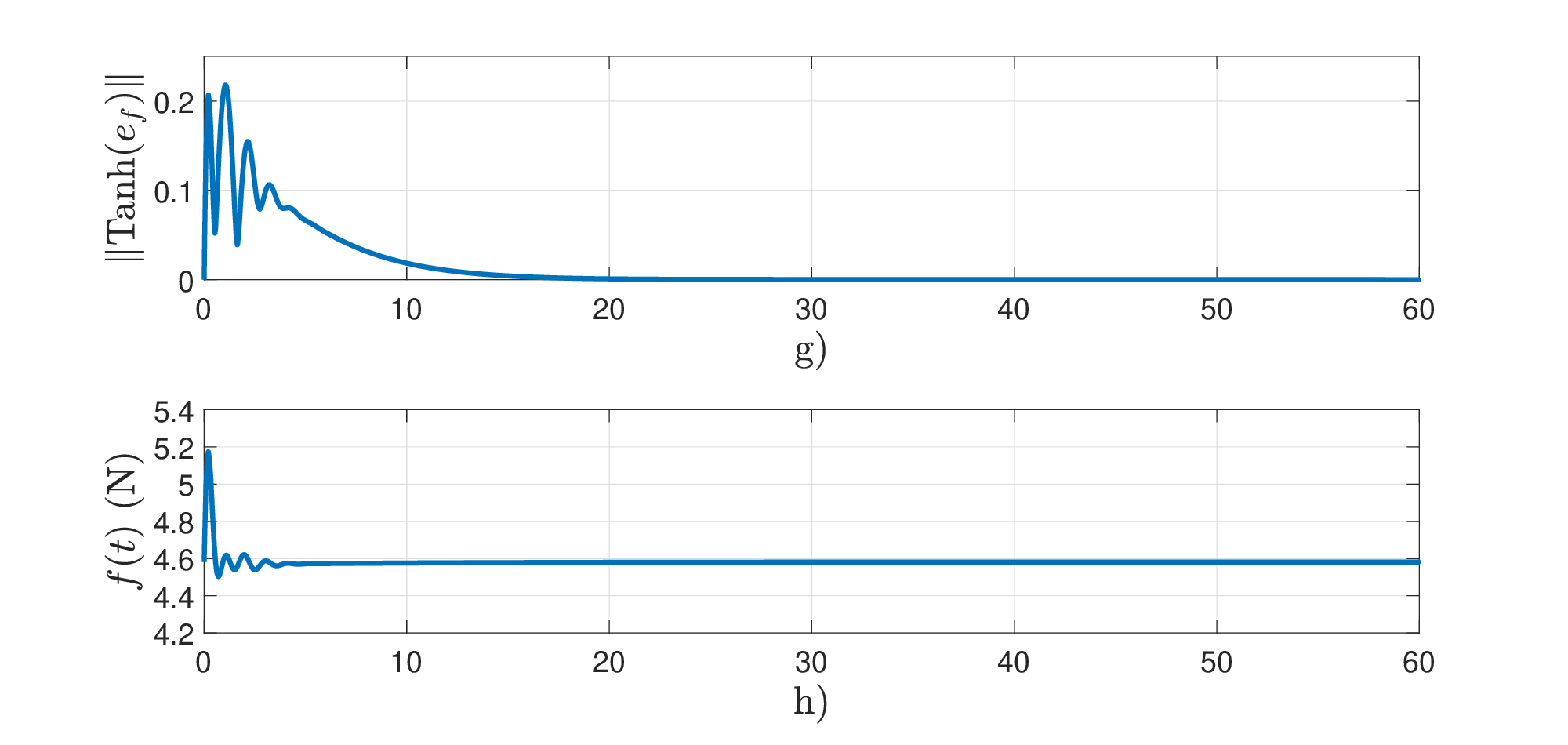}
		\includegraphics[trim = 17mm 80mm 17mm 0mm,clip,scale=0.26]{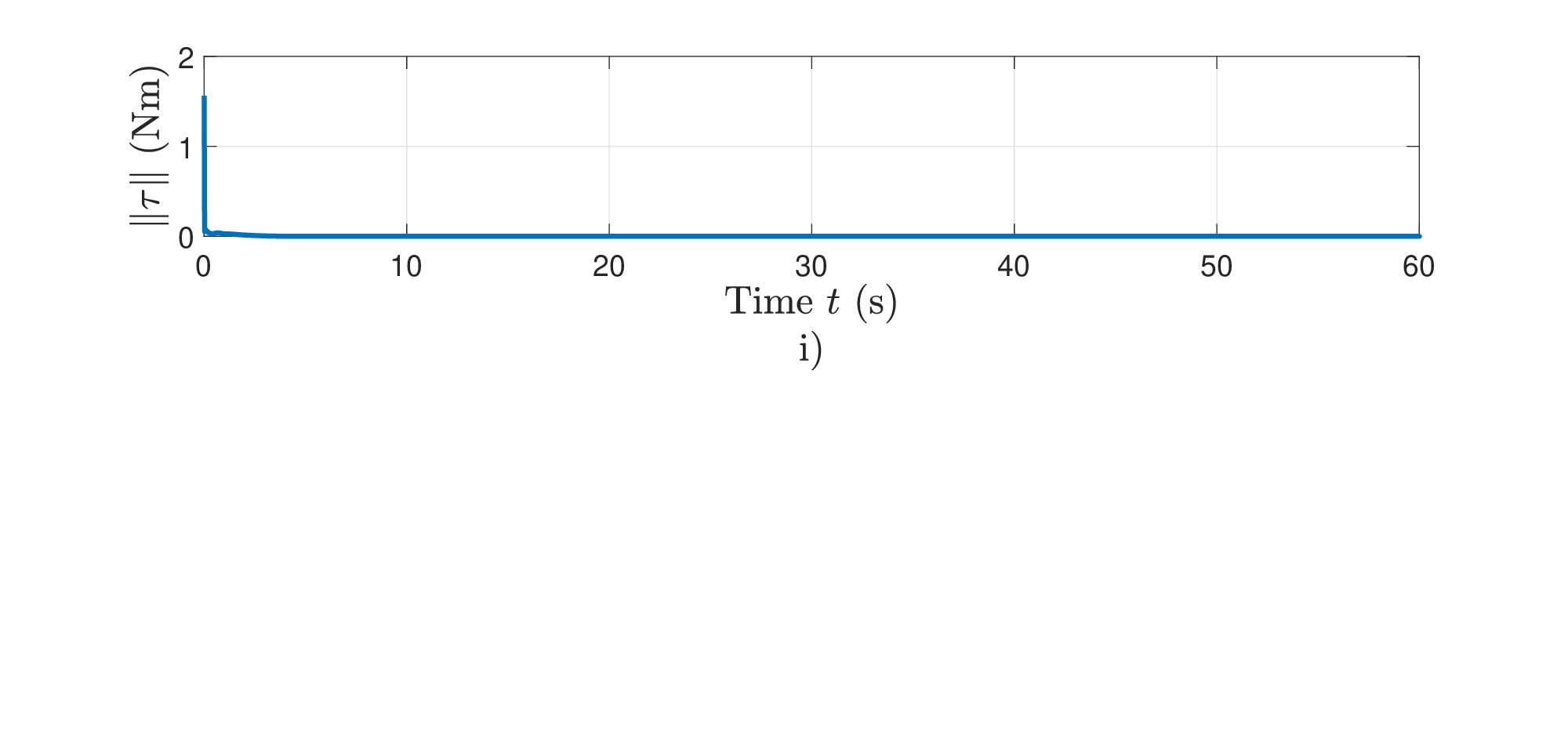}
		\caption{Scenario 3 (uncertain inertia matrix). Performance of the proposed controller under uncertainty in the inertia matrix.}
		\label{fig:UPCase}
\end{figure}

  \subsection{Scenario 4. Apparent acceleration.}

The accelerometer is used to acquire one of the two vector measurements as commented in Remark \ref{rmk6} in this scenario, keeping all the conditions as in Scenario 1, except that the inertial reference vector $r_1$ in Table \ref{tab:Parameters} was replaced by $( ge_z + \dot{v}) / \|ge_z + \dot{v}\|$, i.e., the first vector measurement is $v_{1}=R^{T}( ge_z + \dot{v}) / \|ge_z + \dot{v}\|$. 
Fig. \ref{fig:ApAcc} illustrates the simulation results. As can be observed, there are no substantial differences in controller performance in scenario 1, except for a slight oscillation in the norm $\|\mathrm{Tanh}(e_{f})\|$ after $20$ (s).

\begin{figure}[h]
	\centering
		\includegraphics[trim = 17mm 0mm 17mm 0mm,clip,scale=0.26]{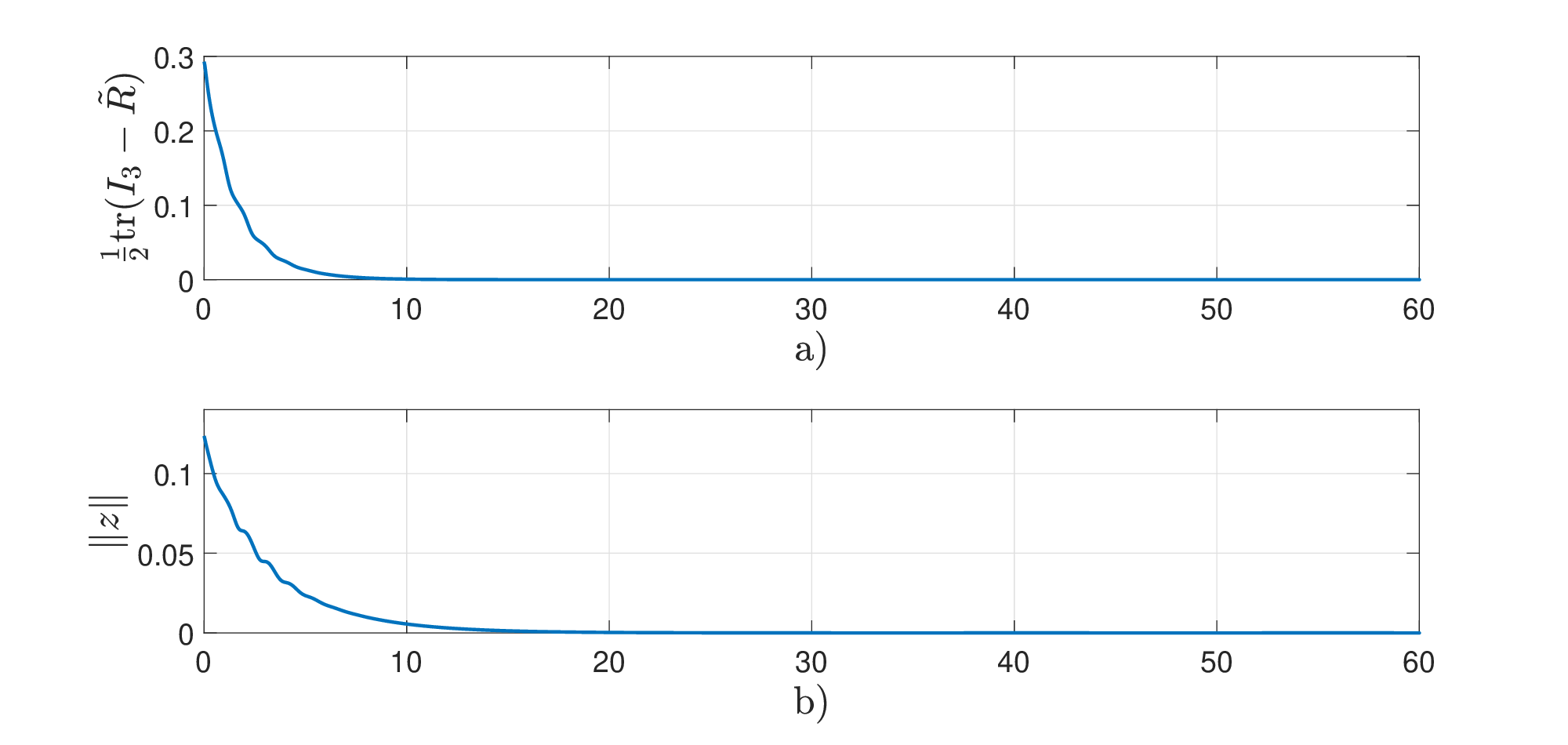}
		\includegraphics[trim = 17mm 0mm 17mm 0mm,clip,scale=0.26]{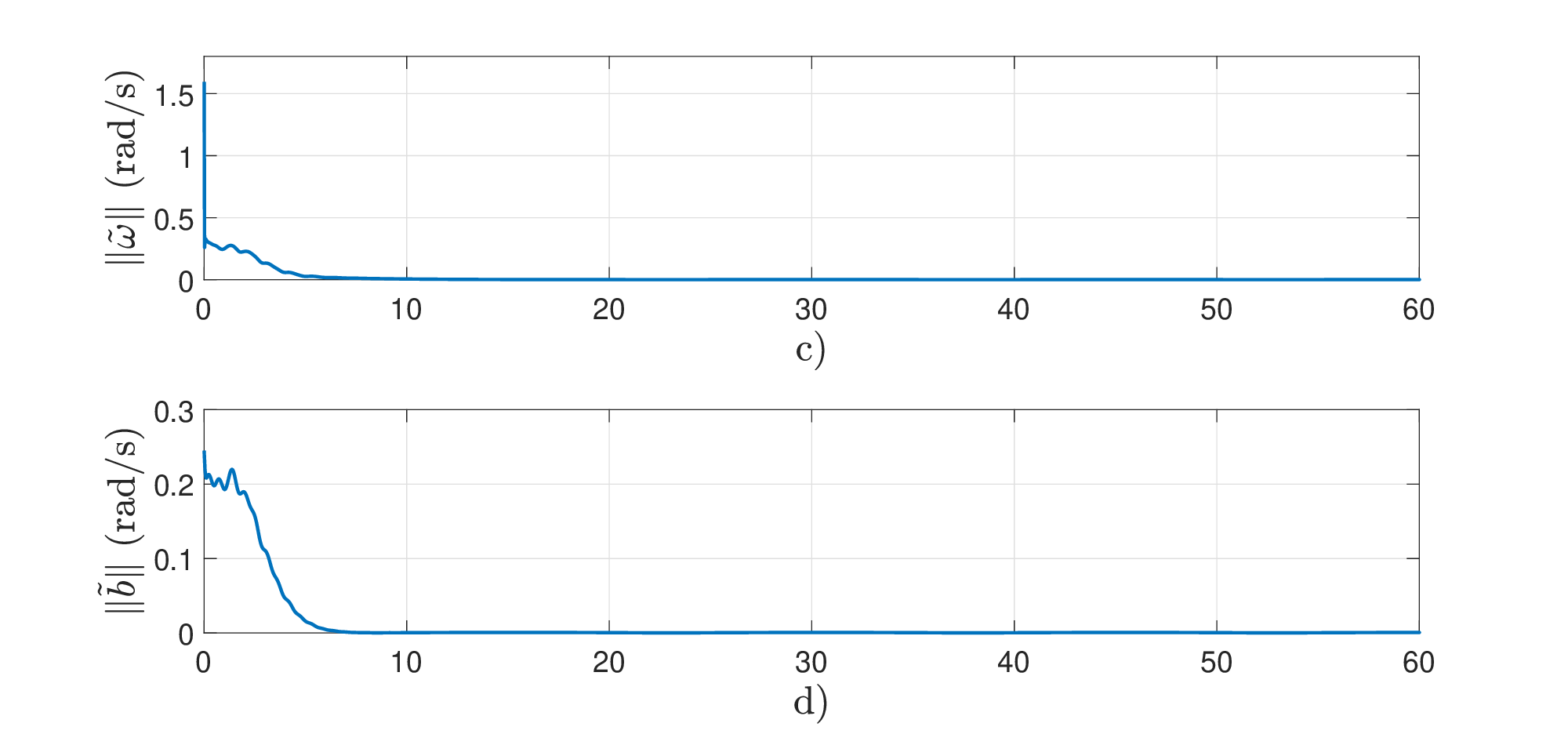}
		\includegraphics[trim = 17mm 0mm 17mm 0mm,clip,scale=0.26]{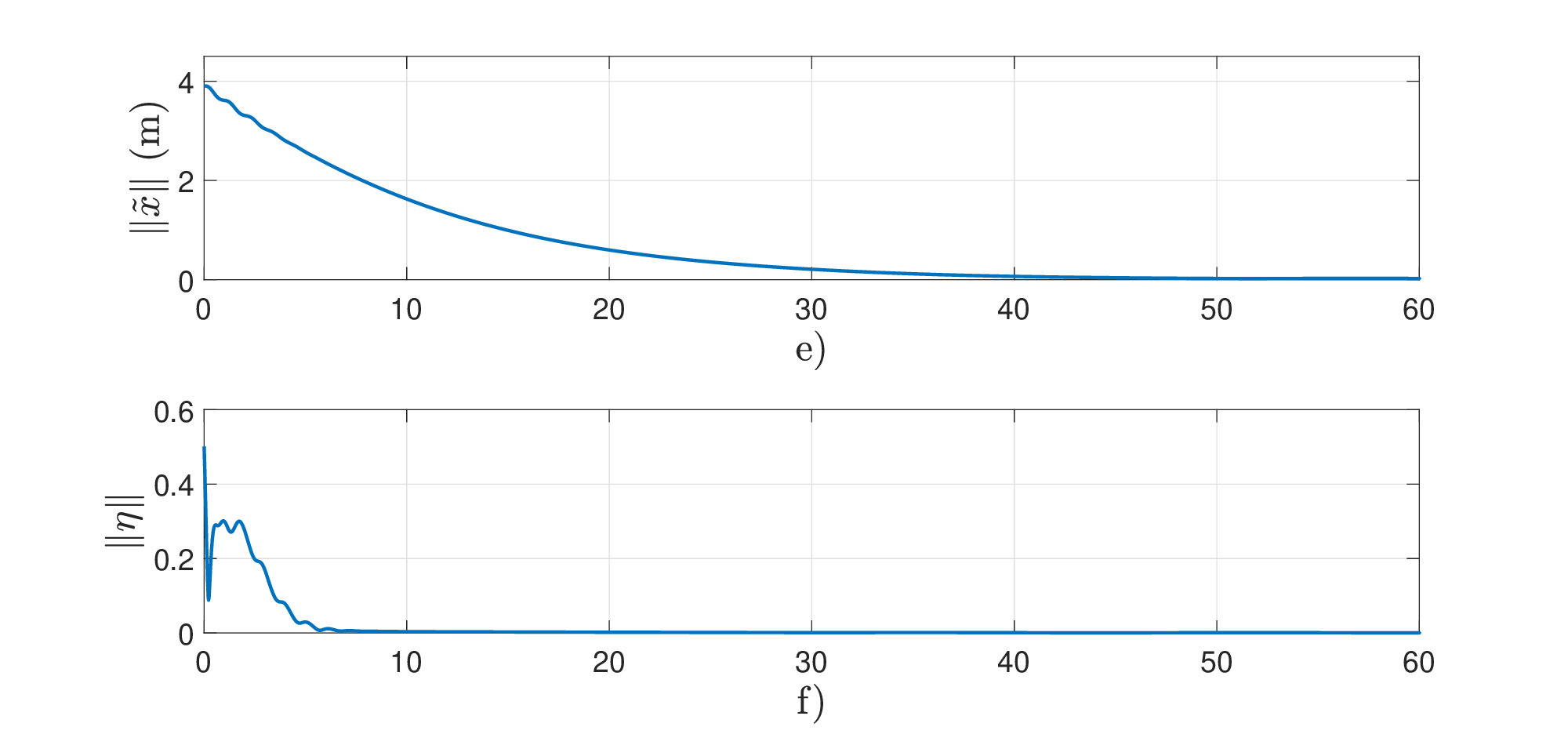}
		\includegraphics[trim = 17mm 0mm 17mm 0mm,clip,scale=0.26]{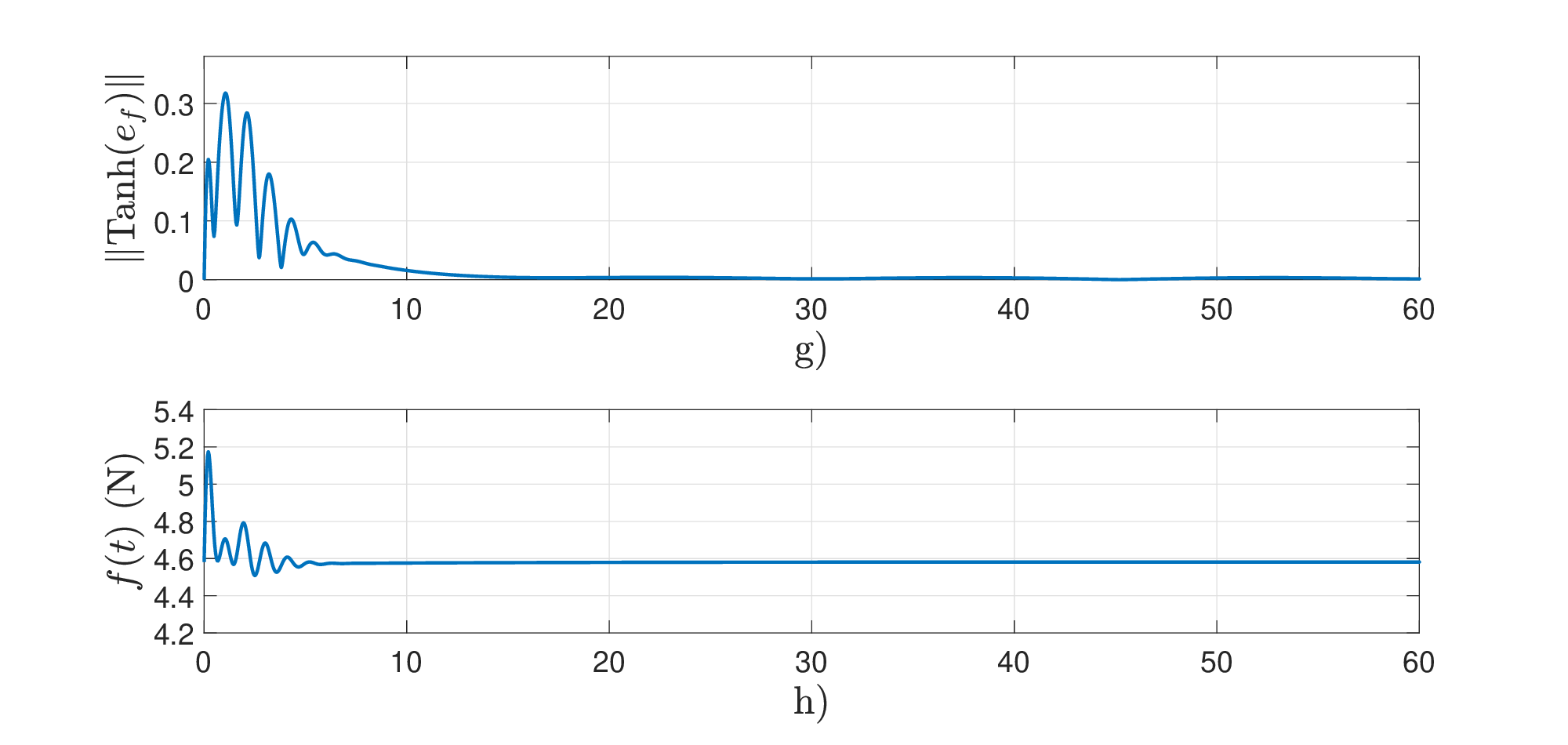}
		\includegraphics[trim = 17mm 80mm 17mm 0mm,clip,scale=0.26]{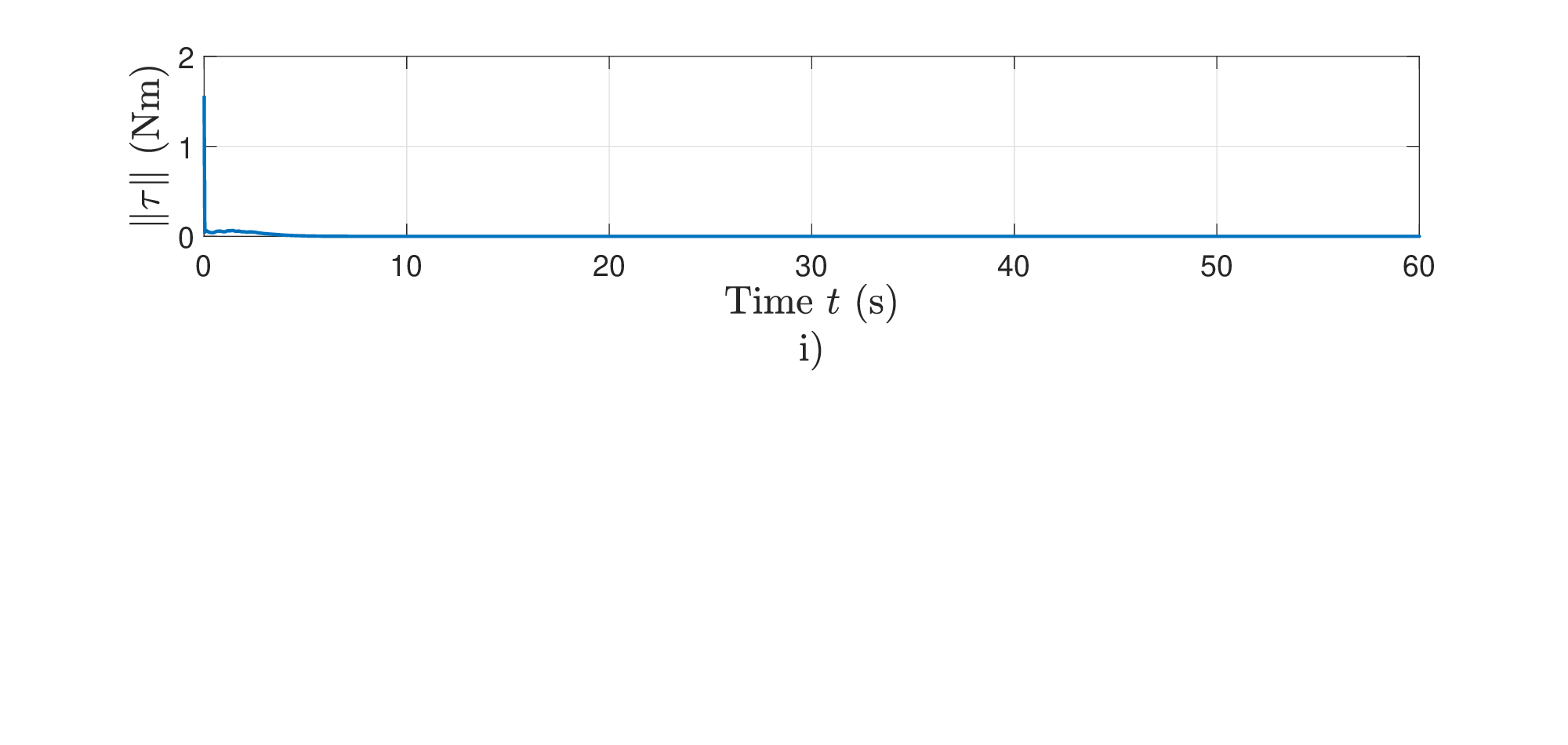}
		\caption{Scenario 4 (Apparent acceleration). Performance of the proposed controller using an accelerometer.}
		\label{fig:ApAcc}
\end{figure}

\section{Conclusion}\label{Sec:Conc}
This paper has designed a tracking control for a quadrotor UAV using inertial vector measurements and gyro rates. No attitude representation or measurement in any form is needed to implement the controller. The semiglobal exponential stability of the overall closed-loop system was demonstrated.  Numerical simulations were included to verify the theoretical results, illustrating the robustness of the proposed controller in the presence of measurement noise, inertia parametric uncertainty, and using apparent acceleration as one of the vector measurements. 

Future work includes adding learning-based modules to model aerodynamic effects and model mismatches to face challenges in real-world applications \citep{song2023reaching}.







\noindent

\bigskip

\appendix
\section{Practical stability}
Consider the system
\begin{equation}\label{eq:system}
    \dot x=f(t,x)+d(t), \ \forall t\geq 0,
\end{equation}
where $f:\mathbb{R}_+\times \mathbb{D}\to \mathbb{R}^n$ is piece-wise continuous in $t$ and locally Lipschitz in $x$, $\mathbb{D}\subseteq \mathbb{R}^n$ is a domain that contains the origin, $f(t,0)=0_{n\times 1}$, and $d(t)$ with $\|d(t)\| \leq \bar{d}$ for some unknown constant $\bar d>0$, representing bounded disturbances. 

\begin{lem}[\bf{Practical stability}]\label{lem-pracStab}
If there exists a continuously differentiable function $V:\mathbb{R}_+\times \mathbb{D}\to \mathbb{R}$ such that
\begin{align*}
    &k_1\|x\|^2\leq V(t,x)\leq k_2 \|x\|^2, \\
    &\frac{\partial V}{\partial t}+\frac{\partial V}{\partial x}f(t,x)\leq -(k_3+\epsilon)
\|x\|^2,\\
&\|\frac{\partial V}{\partial x}\|\leq k_4\|x\|,
\end{align*}
$\forall t\geq 0$ and $\forall x\in \mathbb{D}$, where $k_1,\ k_2, \ k_3,\ k_4$ are positive constants and $\epsilon>0$. Then, 
\begin{enumerate}[(i)]
    \item In the absence of disturbances, i.e., $d(t)=0_{n\times 1}$, the equilibrium $x=0_{n\times 1}$ is {\it exponentially stable}, i.e., 
\begin{equation*}
    \|x(t)\|\leq \big(\frac{k_2}{k_1}\big)^{\frac{1}{2}} \|x(t_0)\| \exp{\big\{-\frac{k_3+\epsilon}{2k_2} (t-t_0)\big\}}, \ \ \forall t\geq t_0, \text{and }\forall x(t_0)\in \mathbb{D}. 
\end{equation*}
\item In the presence of disturbances, the origin $x=0_{n\times 1}$ is {\it practically stable} in the sense that the solution is uniformly ultimately bounded, that is,
\begin{equation*}
    \|x(t)\|\leq \big(\frac{k_2}{k_1}\big)^{\frac{1}{2}} \|x(t_0)\| \exp{\big\{-\frac{k_3}{2k_2} (t-t_0)\big\}}+\frac{(k_4 \bar d)}{2} \sqrt{\frac{k_2}{k_1k_3\epsilon}}, \ \ \forall t\geq t_0, \text{and }\forall x(t_0)\in \mathbb{D}. 
\end{equation*}
\end{enumerate}
If $\mathbb{D}=\mathbb{R}^n\backslash \mathbb{U}$ with $\mathbb{U}$ containing some isolated points other than the origin, then the stability is almost global; If $\mathbb{U}$ is empty, then the stability is global. 
\end{lem}

\begin{proof}
    The time derivative of $V(x,t)$ along the solution of the system \eqref{eq:system} is
\begin{align}
    \dot V(x,t)&=\frac{\partial V}{\partial t}+\frac{\partial V}{\partial x}\Big( f(x,t)+d(t)\Big) \notag\\
    &\leq -(k_3+\epsilon)\|x\|^2+k_4\bar d \|x\| \label{eq:inequality1} \\
    &\leq -k_3\|x\|^2+\frac{(k_4 \bar d)^2}{4\epsilon} \notag \\
    &\leq -\frac{k_3}{k_2}V+\frac{(k_4 \bar d)^2}{4\epsilon}. \label{eq:inequality2}
\end{align}
Therefore,
\begin{equation*}
    V(x,t)\leq \exp{\left\{-\frac{k_3}{k_2}(t-t_0)\right\}} V\big(x(t_0),t_0\big)+\frac{(k_4 \bar d)^2}{4\epsilon}\frac{k_2}{k_3},
\end{equation*}
and 
\begin{equation*}
    \|x(t)\|\leq\big(\frac{k_2}{k_1}\big)^\frac{1}{2}\exp{\left\{-\frac{k_3}{2k_2}(t-t_0)\right\}} \|x(t_0)\|+\frac{(k_4 \bar d)}{2} \sqrt{\frac{k_2}{k_1k_3\epsilon}}.
\end{equation*}
\end{proof}

\section{Proof of Lemma 1-(iv)} \label{App1}
Given $n$ inertial vectors $v_{i}=R^{T}r_{i}$, and their corresponding desired vectors $v_{d,i} = R^{T}_{d}r_{i}$, where $r_{i}\in\mathcal{S}^{2}$ are constant inertial references for $i=1,2,\cdots ,n$. The following matrices can be written
\begin{align}
    H_{B} &= \left[ \sqrt{k_{1}}v_{1} \; \sqrt{k_{2}}v_{2} \;\cdots\; \sqrt{k_{n}}v_{n} \right] \in\mathbb{R}^{3\times n},\notag\\
    H_{D} &= \left[ \sqrt{k_{1}}v_{d,1} \; \sqrt{k_{2}}v_{d,2} \;\cdots\; \sqrt{k_{n}}v_{d,n} \right] \in\mathbb{R}^{3\times n},\notag\\
    H_{I} &= \left[ \sqrt{k_{1}}r_{1} \; \sqrt{k_{2}}r_{2} \;\cdots\; \sqrt{k_{n}}r_{n} \right] \in\mathbb{R}^{3\times n},\notag
\end{align}
then, it is held
\begin{align}
    \bar{W} &= H_{I}H^{T}_{I},\label{eq:WbarH}\\
    \varepsilon &= \sum^{n}_{i=1}k_{i}\left(1-v^{T}_{i}v_{d,i}\right) = \frac{1}{2}\sum^{n}_{i=1}k_{i}\|v_{i} - v_{d_{1}}\|^{2} \notag\\
    &= \quad \frac{1}{2}\mathrm{tr}\left(\left( H_{B} - H_{D}\right)^{T}\left( H_{B} - H_{D}\right)\right) . \label{eq:varEpsH}
\end{align}

Furthermore, being $H_{B}-H_{D} = (R^{T}-R^{T}_{d})H_{I}$, it can be solved for $R^{T}-R^{T}_{d}$, by using \eqref{eq:WbarH}, as follows
\begin{align}
    \left( H_{B}-H_{D}\right)H^{T}_{I} &= (R^{T}-R^{T}_{d})H_{I}H^{T}_{I} \notag\\
    &= (R^{T}-R^{T}_{d})\bar{W}, \notag \\
    \left( H_{B}-H_{D}\right)H^{T}_{I}\bar{W}^{-1} &= R^{T}-R^{T}_{d}, \label{eq:RRdT}
\end{align}
later, in view of \eqref{eq:WbarH}-\eqref{eq:RRdT}, it can be written
\begin{align}
    \|R-R_{d}\| &\leq \|R-R_{d}\|_{F} = \|R^{T}-R^{T}_{d}\|_{F} \notag\\
    &\leq \| H_{B}-H_{D}\|_{F}\|H^{T}_{I}\bar{W}^{-1}\|_{F} \notag\\
    &= \sqrt{\mathrm{tr}\left(\left( H_{B} - H_{D}\right)^{T}\left( H_{B} - H_{D}\right)\right)} \sqrt{\mathrm{tr}\left( \bar{W}^{-T}H_{I}H^{T}_{I}\bar{W}^{-1} \right)} \notag\\
    &=\sqrt{2\varepsilon}\sqrt{\mathrm{tr}\left( \bar{W}^{-T}\right)}=\sqrt{2\varepsilon \varpi} ,\notag 
\end{align}
where the fact that $\mathrm{tr}(A)=\mathrm{tr}(A^{T})$, and $\varpi \vcentcolon = \mathrm{tr}\left( \bar{W}^{-1}\right)>0$ were used. Therefore, according to \eqref{eq:bcond}, it yields
\begin{equation*}
    \|R-R_{d}\| \leq \sqrt{\frac{\varpi \beta}{\alpha_{1}}}\|z\|.
\end{equation*}

\bibliographystyle{apalike}
\bibliography{references}  

\end{document}